\documentclass[9pt,shortpaper,twoside,web]{ieeecolor}
\usepackage{apg}
\makeatletter\let\NAT@parse\undefined\makeatother 
\usepackage[numbers]{natbib}
\usepackage{amsmath,amssymb,amsfonts}
\usepackage{algorithmic}
\usepackage{graphicx}
\usepackage{textcomp}
\def\BibTeX{{\rm B\kern-.05em{\sc i\kern-.025em b}\kern-.08em
    T\kern-.1667em\lower.7ex\hbox{E}\kern-.125emX}}
\usepackage[dvipsnames]{xcolor}
\usepackage[switch]{lineno}
\usepackage{subcaption}
\usepackage{booktabs,multirow,accents}
\DeclareRobustCommand{\munderbar}[1]{ \underaccent{\bar}{#1}}
\usepackage{bm}\renewcommand{\vec}{\bm}
  
\usepackage{amsthm}
\newtheorem{thm}{Theorem}[section]
\newtheorem{corollary}{Corollary}[thm]

\newtheorem{prop}[thm]{Proposition}

\theoremstyle{definition}

\theoremstyle{remark}

\makeatletter\def\operator@font{\sf}\makeatother

\usepackage{lastpage,fancyhdr}
\usepackage{ifthen,hyperref}
\fancyheadoffset[L]{0.25in}
\fancyfootoffset[L]{0.25in}
\title{Insect-inspired Visually-guided Decentralized Swarming}
\author{Mehdi Yadipour, and Imraan A. Faruque
\thanks{M.~Yadipour is with Oklahoma State University, Stillwater, OK 74078 USA (e-mail: mehdi.yadipour@okstate.edu). }
\thanks{I.~A.~Faruque is with Oklahoma State University, Stillwater, OK 74078 USA (e-mail: i.faruque@okstate.edu).}
\thanks{This work was supported in part by ONR Young Investigator Award N00014-19-1-2216. Copyright by the authors June 2022; this work may be under consideration for publication and copyright may be transferred without notice, after which this version may no longer be accessible.}
}

\begin{document}
\maketitle

\begin{abstract}This paper addresses the need for fast, lightweight, vision-guided swarming under limited computation and no explicit communication network or position source. The study develops a multi-agent optic flow sensing framework, then integrates ``perfect information'' distributed feedback with optic flow sensing to create an analogous visually-guided feedback path for idealized inter-agent velocity and distance structures. The Cucker-Smale flocking example is used to develop vision-guided swarming with rigorous asymptotic convergence guarantees, including under ignorance of agent size.\end{abstract}
\section{Introduction}\label{s:Introduction}
Achieving efficient sensing and feedback paths to support aerial swarming has been a persistent challenge to aerospace robotics and motion control. Insects are an example of biological systems able to achieve robust inflight coordination with limited sensing and limited computational resources.  The dominance of visual feedback in such animals shows a bias towards visual interactions \cite{parsons2010sensorFusion,taylor2007whatMeasureAndWhy}, and visual feedback signals have been shown to be important for their solitary \cite{humbert2009visuomotorConvergence} and group navigation \cite{billah2020multiAgentWfi}.

Nonetheless, there remains a disconnect between the understanding of insect sensing and feedback rules and swarming theory. Many swarming tools rely on relatively large amounts of information (instantaneous position and velocity of all other agents), often shared through an explicit communication network. Restricting information flow significantly affects algorithm design and performance \cite{passino2006InformationFlow,bulent2015partialInfoNashEqb}, a challenge exacerbated in fast-moving dynamic systems like high speed aerial vehicles. Moreover, many continuous time swarm motion models lead to coordinated group motion having relatively high polarization levels, more analogous to bird flocking and fish schooling. In this regard, the apparently more chaotic motions seen in crowded assemblies of flying insects are not fully explained or translated to control theoretic framework.

This paper introduces a novel visually-guided feedback law consistent with insect visual sensing structures that is shown rigorously to converge to coordinated group motion. The visual feedback rule is shown to require no explicit communication network or information sharing and require no explicit position reference. Further proofs and illustrative simulations show that it is robust to ignorance of agent size, measurement noise, and generates a more chaotic swarming motion than the seminal Cucker-Smale swarm theoretic formulation.

The main contribution of this paper is to show (by construction) that  a concise output feedback on optic flow, assisted by heading rate and quadrant-level heading knowledge, is sufficient to produce coordinated motion. The result is is flexible to agent number and size.

\section{Background and previous work}
\label{s:Background and previous work}
\subsection{Swarm theory and proofs}
Swarm theory has benefited from considerable academic attention, and a balanced review is beyond the scope of this paper; instead the reader is referred to \citet{chung2018SwarmSurveyTac}, \citet{rossi2018review}, and \citet{olfatisaber2006Review}. This treatment will focus on a particular structure that will be shown to be amenable to rigorous visually guided swarming in aerial vehicles without communication networks, especially for individual models implementing onboard distributed feedback. This structure comprises an individual feedback law where an agent updates its velocity $v_i$ as a function of inter-agent relative velocities ${v}$ and distances $r$, (i.e., in the form $\dot{v}_i=f(v,r)$). Several of the most well-received flocking models may be written in this form, e.g., \citet{reynolds1987flocksBoids,vicsek1995swarm}, and \citet{cucker2007TacEmergent}. We will refer to these as $(v,r)$ models to denote their dependence on velocity and radial distance and use the Cucker-Smale formulation as a canonical example.

\citet{cucker2007TacEmergent} developed a dynamic model representing individual agent dynamics in a group of flocking agents such as birds. The dynamics describe an agent $i$'s velocity $v_i$ as evolving according to relative velocity and relative distance between agent $i$ and other agents of the group.
The primary result is a theorem showing the agent velocities converge to a common magnitude, or flocking. The C-S model's resemblance to Newtonian physics and rigorous asymptotic convergence proof inspired a large number of variants \cite{Choi2017csModelAndVariants}.

While foundational, the original Cucker-Smale framework assumes perfect information, i.e., that agents have access to noise free measurements of all agents' positions and velocities, which is not achievable in real systems having measurement noise. \citet{Cucker2008mordeckiNoise} initially incorporated additive and multiplicative noise on velocity, however most subsequent treatments use stochastic analysis \citep{Jabin2017meanFieldStochastic} and Gaussian/Wiener process frameworks \citep{Park1981wienerProcess} to focus on multiplicative velocity noise only. 
Progress towards visually-guided flocking has generally used vision in concert with other sensors \citep{2019yazheVisionAidedFlocking}. Early work on flocking ground robots used overhead motion tracking systems as a sensing proxy to test visually guided algorithms \citep{moshtagh2007flockingTac}, and subsequent work showed the sensitivity of visually-guided flocking via an individualistic framework to optical and occlusal characteristics \citep{dachner2018humanVisualCoupling}, which affect swarm performance \citep{asadi2016limitedFovSwarm,soria2019LatVisionFlocking}.
Sensor modality treatments suggest the modality that best supports the high bandwidth relative navigation needed for neighbor-relative flight motions are their visual domain sensors \cite{taylor2007whatMeasureAndWhy}, largely composed of compound eyes and rudimentary eyes (ocelli) providing relatively noisy information. These sensors show a dominance of sensitivity to optic flow, or visual motion blur,  \citep{serres2017opticFlowCollision} with specialized neurons (e.g., lobula plate tangential cells) responsive to optic flow patterns \cite{krapp1996ofMeasurement} (defined as $\dot{Q}$ later in this paper). Insects also demonstrate an ability to resolve small targets \citep{nordstrom2009stmdHypercomplex,nordstrom2012stmdsNeuralSpecializations} and to make dynamic computations on these small targets \citep{wiederman2017predictiveFocusGainModulation}. Futhermore, insects have demonstrated an internal compass, likely supported by their toroid-shaped ellipsoid bodies and protocerebral bridge structures \citep{seelig2015neuralStructureCompass}, that encodes direction from solar/lunar signals to wind \cite{zittrell2020PolarizationSunCompassLocust,dacke2021insectDrosophilaCompass}. These sensory structures and their engineering analogues have a rich history of performance in solitary environments on both insects and robots \cite{srinivasan1992beesOpticFlow,escobar2019smallObjectAvoidance,humbert2009visuomotorConvergence}, and a growing understanding of performance in multi-agent contexts, where they can support coordinated group motions \citep{billah2020multiAgentWfi,moshtaghcross,billah2022stmdMultiAgent}.

\subsection{This paper} 
Despite the progress in swarm theory, visually guided flight, and insect work, the systems and control field does not yet include a visual sensing analogue of the C-S result, or for more general $(v,r)$ models, which we use to indicate a dynamic model in which individual agent states states are updated by velocity ($v$) and distance ($r$) terms. The primary contribution of this work is to develop the theory to enable visually-guided flocking by integrating a visual navigation framework with the $(v,r)$ flocking model, using Cucker-Smale flocking as an example $(v,r)$ model. We consider it plausible that an insect (or robot), in general an agent, can measure optic flow field surrounding it and heading information, perhaps with some noise. We do not presume they possess an inflight communication network with sufficient bandwidth to coordinate flight or any position reference.

A second contribution of this work is to provide a mechanism for a multi-agent system to utilize the optic flow fields perceived by the agents and achieve speed consensus, in addition to orientation consensus (e.g., see \citep{moshtagh2007flockingTac}). As an example, the optic flow field perceived by a bee from a group of bees flying inside a tunnel is computed and illustrated as a heat map in Fig.~\ref{f:F_alpha}. Signals corresponding to a given neighbor agent have comparable magnitudes, suggesting that insects with limited neural capacities and thus unable to support engineered feature extraction may be able to use the number of signals having approximately equal magnitudes as an approach to estimating figure size. This view of how optic flow arises from surrounding agent motions supports the idea that a given region's optic flow signals serve to quantify relative distance (through the number of comparable magnitude signals) and relative speed (through the magnitude itself.

\begin{figure}[htbp]\centering
\includegraphics[width=0.35\textwidth]{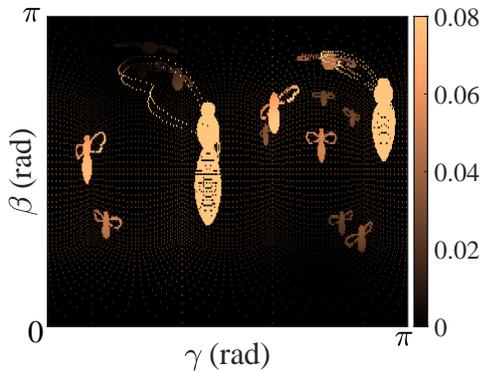}   
\caption{Optic flow field visualization: A bee (viewer) is assumed to be flying together with a group of bees. The optic flow signals received by the viewer at its eyes' image plane are computed in the agent's body frame using azimuth-elevation coordinates  at one degree resolution; color indicates the magnitude of these optic flow signals.}\label{f:F_alpha}\end{figure}

\section{Formulation \& Analysis}
In this section, we develop a formulation of dynamic agents connected by only optic flow signals, and show that it leads to velocity (speed and heading) convergence, and is well-behaved under varying agent size and measurement noise.

\subsection{Formulation}
The number of pixels on a sensor's image plane or the number of optic flow receptive neurons in the third optic ganglion (lobula plate) of a flying insect's eye is bounded, hence we assume that the optic flow field is a finite and countable set.
Consider the two arbitrary agents $i$ and $j$ illustrated in Fig.~\ref{f:geometry_1}, which both belong to a group of $N$ agents. We assume the motions of all $N$ agents are restricted to be on an inertial ($x-y$) plane and denote the distance between agent $i$ and $j$ as $r_{ij}$. Without loss of generality, denote as $L$ the agent (insect) semi-length or semi-height (a later theorem will show no loss of generality). 
Let $n_{ij}$ denote the number of optic flow signals $\dot{Q}_{ij}$ received by agent $i$ due to relative motion between agents $i$ and $j$. This number is proportional to $r_{ij}$ and its time derivative $\dot{n}_{ij}$ proportional to the rate of distance change between agents $\dot{r}_{ij}$. The numbers $n_{ij}$ and $\dot{n}_{ij}$ can then be used as a proxy for $r_{ij}$ and $\dot{r}_{ij}$, respectively. Additionally, the number $n_{ij}$ is proportional to the angle $\alpha_{ij}$ Fig.~\ref{f:geometry_1}. Therefore, without loss of generality, $\alpha_{ij}$ and $\dot{\alpha}_{ij}$ will be used as a measure of $n_{ij}$ and $\dot{n}_{ij}$,  respectively.  

\begin{figure}[htbp]\centering
\includegraphics[width=0.35\textwidth]{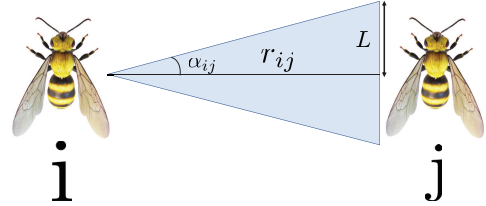}   
\caption{The agent $j$ of size $L$ located at distance $r_{ij}$ from agent $i$ is one of the contributors to the optic flow signals received by agent $i$. Angle $\alpha_{ij}$ can be used as a measure of number of optic flow signals received by agent $i$ due to relative motion between agents $i$ and $j$.}
\label{f:geometry_1}
\end{figure}

Introducing a (non-holonomic) no-sideslip dynamic model for each agent, the individual agent dynamics are given as
\begin{equation}\label{e:Agent_Dynamics}
\begin{split}
\dot{v}_i&=u^v_i\\
\dot{y}_i&=v_i\sin{\theta}_i\\
\dot{x}_i&=v_i\cos{\theta}_i\\
\dot{\theta}_i&=\omega_i\\
\dot{\omega}_i&=u^{\omega}_i,\end{split}\end{equation}
in which orientation $\theta_i$, speed $v_i$, and coordinates $x$ and $y$ are the states of agent $i$ with respect to inertial frame. The control inputs $u^i_{\theta}$ and $u^i_v$ act on the agent's orientation and speed, respectively. 

The foundational flocking result found in \citet{cucker2007TacEmergent} will be used here to represent the ``ideal sensing" limit, or the idealized flocking when all other agents instantaneous positions and velocities are known with no measurement noise or delay. 
\begin{thm} 
Cucker-Smale flocking postulates that an agent $i$ updates its velocity $\vec{v}_i$ while in a group of $N$ agents using the rule
\begin{equation}\label{e:Continuous_CS_Dynamics}
\dot{\vec{v}}_i=\sum_{j=1}^N \frac{H(\vec{v_j}-\vec{v_i})}{\left[\sigma^2+r_{ij}^2\right]^{\beta}},\end{equation}
where $H$, $\sigma$, and $\beta$ are positive constants and $r_{ij}=\lVert \vec{r_j}-\vec{r_i}\rVert$ is the distance between agent $i$ and $j$ of the group. Then, for $\beta<\frac{1}{2}$, when $t\rightarrow \infty$ the velocities $v_i(t)$ tend to a common limit $\hat{v}\in\mathrm{E}^3$ and the vectors $x_i-x_j$ tend to a limit vector $\hat{x}_{ij}$, for all $i,j\leq k$. The same holds for $\beta\geq\frac{1}{2}$ provided the initial conditions (initial position $x(t=0)$ and velocity $v(t=0)$) satisfy an additional constraint.
 \label{l:csThm} \end{thm}
\begin{proof}Eqn.~\eqref{e:Continuous_CS_Dynamics} rewrites the Cucker Smale dynamics \citep{cucker2007TacEmergent} in $(v,r)$ form, thus Thm.~\ref{l:csThm} is established by \citet{cucker2007TacEmergent}.\end{proof}

\subsubsection{Multi-agent optic flow-feedback}
We now introduce the multi-agent optic flow feedback model and its corresponding convergence theorem as the main theorem of this paper. Geometry is based on frames and angle definitions in Fig.~\ref{f:geometry_2} 
\begin{thm}[YFM feedback] \label{t:OfFlockingThm}
Consider a collection of agents in which each agent $i$ updates its speed $v_i$ and heading $\theta_i$ using feedback laws based on optic flow regions $Q_{ij}$ induced by neighboring agents $j$ using the following two update rules.

Suppose agents $i=1,...N$ update their speeds using
\begin{equation} \dot{v}_i\!=\!\tiny HL\sum_{j=1}^N\normalsize\frac{\mp\dot{\alpha}_{ij}(1\!+\!\cot^2{\alpha_{ij})}\cos\gamma_{ij}-(\dot{Q}_{ij}\!+\!\dot{\theta_i})\cot\alpha_{ij}\sin\gamma_{ij}}{\left(1+L^2\cot^2\alpha_{ij}\right)^{\beta}}
\label{e:vdot}
\end{equation} 
and update their orientations using 
\begin{equation}\begin{split}
\ddot{\theta}_i=-k\dot{\theta}_i+\hspace{.4\textwidth}\\
\frac{kHL}{v_i}\sum_{j=1}^N\frac{\mp\dot{\alpha}_{ij}(1\!+\!\cot^2{\alpha_{ij})}\sin\gamma_{ij}+\left(\dot{Q}_{ij}\!+\!\dot{\theta}_i\right)\cot\alpha_{ij}\cos\gamma_{ij}}{\left(1+L^2\cot^2\alpha_{ij}\right)^{\beta}},\end{split}
\label{e:thetaDot}
\end{equation}
where $k>0$ denotes is a feedback gain, $H$ is a positive constant, $\gamma_{ij}\in\left[-\pi,\pi\right]$ is the observation angle of target $j$ with respect to a frame affixed to agent $i$ as shown in Fig.~\ref{f:geometry_2}, 
and the sign $\mp$ updates as \[\mp= \left\{\begin{array}{ll}
 -, &\mathrm{if~} \lvert\theta_i+\gamma_{ij}\rvert \in [0,\pi/2]\\
 +, &\mathrm{if~} \lvert\theta_i+\gamma_{ij}\rvert \notin [0,\pi/2]\,.\end{array}\right. \]
 
Then, for this collection of agents, the speeds $v_i$ and headings $\theta_i$ both converge asymptotically to common values if $\beta<\frac{1}{2}$.  If $\beta\geq\frac{1}{2}$, convergence depends on initial conditions of the agents.\label{mainOfThm}
 \end{thm}

\begin{figure}[h]\centering
\includegraphics[width=0.5\textwidth]{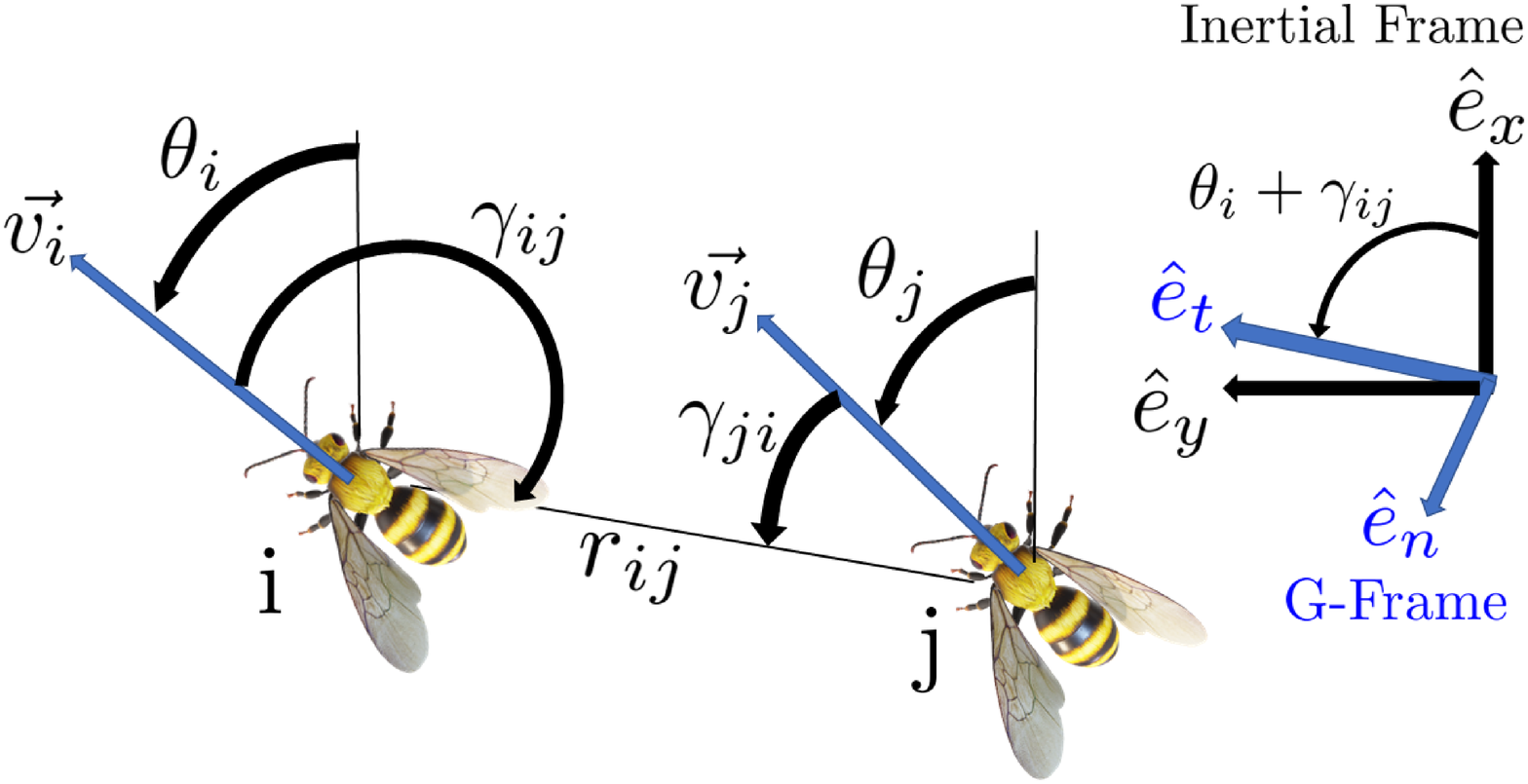}     
\caption{Geometric variables to calculate the optic flow signal between agents $i$ and $j$ located at a relative distance $r_{ij}$ are shown. These agents are moving on a plane with velocities $\vec{v_{[\:]}}$ at directions $\theta_{[\:]}$ and viewing each other at angles $\gamma_{[\:,\:]}$. Angles are defined in $[-\pi,\pi]$.}
\label{f:geometry_2}  
\end{figure}

Note that the velocity of other agents and distance to other agents are unknown to agent $i$, and the only known information to the agent $i$ is the optic flow field produced by relative motion between agents and heading information.  

\begin{proof}
Generally, the proof consists of (a) writing C-S dynamics in speed/direction from, (b) re-writing optic flow in a new ``sin-sin" form, (c) writing the relative velocity in observation frame, (d) writing relative velocity in terms of visual signals, (e) transforming from body frame into inertial frame, (f) obtaining desired values for speed and orientation update rules, and (g) designing controllers using the desired values.

\paragraph{C-S model in speed/direction form} 
Updating the agent's velocity vector $\vec{v}_i$ by its derivative $\dot{\vec{v}}_i$ may also be written as updating the agent's speed $v_i$ and travel direction $\theta_i$ via their time derivatives $\dot{v}_i$ and $\dot{\theta}_i$. If we write the velocity of agent $i$ as $\vec{v}_i=v^i_{x}\hat{e}_x+v^i_{y}\hat{e}_y$, and its derivative, obtained from C-S model, as $\dot{\vec{v}}_i=\dot{v}^i_x \hat{e}_x+\dot{v}^i_y \hat{e}_y$, then the time derivatives of the speed and orientation of agent $i$ can be calculated by
\begin{equation}\label{e:V_dot}
\dot{v}_i=\frac{v^i_{x}\dot{v}^i_{x}+v^i_{y}\dot{v}^i_{y}}{v_i}
\end{equation}
and 
\begin{equation}\label{e:Theta_dot}
\dot{\theta}_i=\frac{1}{1+\tan^2(\theta_i)}\frac{v^i_{x}\dot{v}^i_{y}-v^i_{y}\dot{v}^i_{x}}{v^2_x}.\end{equation}

In other words, if each agent $i$ updates its speed and orientations using Eqns.~\eqref{e:V_dot} and \eqref{e:Theta_dot}, then by the primary C-S theorem in Thm.~\ref{l:csThm}, convergence is also guaranteed. We denote the desired value of speed change $\dot{v}$ by $\dot{v}^*$ and the desired value of heading change $\dot{\theta}$ by $\dot{\theta^*}$.

\paragraph{Multi-agent optic flow in sin-sin form}
For the agents $i$ and $j$ from a group of $N$ agents let $\vec{v}_i$ and $\vec{v}_j$ be their velocity vectors, respectively. The optic flow $\dot{Q}_{ij}$ for agent $i$ due to relative motion between agents $i$ and $j$ is given by \cite{koenderink1987opticFlowFacts}
\begin{equation}\label{e:qdot}
\dot{Q}_{ij}=-\dot{\theta}_i+\frac{1}{r_{ij}}\left[v_i \sin(\gamma_{ij})-v_j \sin(\gamma_{ji})\right].
\end{equation}

\paragraph{Relative velocity in observation frame} Denote as $\gamma_{ij}$ the viewing angle of agent $i$ as seen by $j$, and $\gamma_{ji}$ to be the viewing angle) of $j$ as seen by $j$ WRT to $i$ as shown in Fig.~\ref{f:geometry_2}. 
Let G-frame $G=\{\hat{e}_t,\hat{e}_n\}$ be defined by two orthogonal unit vectors, $\hat{e}_t$ a unit vector along the $\gamma$ direction (aligned with the target), and $\hat{e}_n$, which is perpendicular to $\gamma$ direction. The relative velocity between agents $i$ and $j$ can be written as
\begin{equation}\label{e:deltaV}
\vec{v}_j-\vec{v}_i= \left(v_j \cos\gamma_{ji}-v_i \cos\gamma_{ij}\right)\hat{e}_t+\left(v_j \sin\gamma_{ji}-v_i \sin\gamma_{ij}\right)\hat{e}_n.\end{equation}
\paragraph{Rewrite velocity in visual signals} We proceed component wise through this expression.
The component acting along the normal vector $\hat{e}_n$ in Eqn.~\eqref{e:deltaV} can be rewritten using the optic flow equation Eqn.~\eqref{e:qdot} as
\begin{equation}\label{e:VjSin-ViSin}
v_j \sin\gamma_{ji}-v_i \sin\gamma_{ij}=-r_{ij}\left(\dot{Q}_{ij}+\dot{\theta}_i\right)\end{equation}

The component of velocity in the target direction $\hat{e}_t$ may be rewritten as a function of the expansion of the angle $\alpha_{ij}$. From Fig.~\ref{f:geometry_1},
\begin{equation}\label{e:rij-LcotAlpha}
r_{ij}=L \cot(\alpha_{ij}).\end{equation}

\noindent Let $\dot{r}_{ij}$ be the time derivative of the distance between agents $i$ and $j$. Then the velocity component along $\hat{e}_t$ in Eqn.~\eqref{e:deltaV} may be written as 
\[v_j \cos(\gamma_{ji})-v_i \cos(\gamma_{ij})=
\left \{ \begin{array}{c c}
 ~\dot{r}_{ij},   & \mathrm{if~} \lvert\theta_i+\gamma_{ij}\rvert \in [0,\pi/2] \\
-\dot{r}_{ij},   &  \mathrm{if~} \lvert\theta_i+\gamma_{ij}\rvert \notin [0,\pi/2].
\end{array} \right.\] From Eqn.~\eqref{e:rij-LcotAlpha} we have $\dot{r}_{ij}=-[1+\cot^2\alpha_{ij}]L\dot{\alpha}_{ij}$. So, the component of the relative velocity along $\hat{e}_t$ can be written as
\begin{equation}
\label{e:VjCos-ViCos}
v_j \cos(\gamma_{ji})-v_i\cos(\gamma_{ij})=
\end{equation}
\[
\left\{
\begin{array}{cl}
   -\left(1+\cot^2\alpha_{ij}\right)L\dot{\alpha}_{ij}  &\mathrm{if~} \lvert\theta_i+\gamma_{ij}\rvert \in [0,\pi/2] \\
   ~~\,\left(1+\cot^2\alpha_{ij}\right)L\dot{\alpha}_{ij} &\mathrm{if~} \lvert\theta_i+\gamma_{ij}\rvert \notin [0,\pi/2].
\end{array}\right.\]
Using Eqns.~\eqref{e:VjSin-ViSin} and \eqref{e:VjCos-ViCos}, we can rewrite the relative velocity between agents $i$ and $j$ in Eqn.~\eqref{e:deltaV} as
\begin{equation}\label{e:relativeVelocityAlphadotQdotForm}
\vec{v}_j-\vec{v}_i= \mp\left(1+\cot^2\alpha_{ij}\right)L\dot{\alpha}_{ij}\hat{e}_t-r_{ij}\left(\dot{Q}_{ij}+\dot{\theta}_i)\right)\hat{e}_n.
\end{equation}

\paragraph{Transform into inertial frame} Let $\hat{e}_x$ and $\hat{e}_y$ be the unit vectors of the inertial $x-y$ frame. The transformation of vectors between inertial frame vectors and $\gamma$-frame can be achieved using

\begin{equation} \label{e:Inertial_to_Gamma}
\begin{bmatrix}\hat{e}_t\\\hat{e}_n\end{bmatrix} = \begin{bmatrix}\cos(\gamma_{ij}+\theta_i) & \sin(\gamma_{ij}+\theta_i)\\ \sin(\gamma_{ij}+\theta_i) & -\cos(\gamma_{ij}+\theta_i) \end{bmatrix} \begin{bmatrix}\hat{e}_x\\ \hat{e}_y\end{bmatrix}.\end{equation}
Defining $\phi_{ij}=\gamma_{ij}+\theta_i$ for compact notation, the relative velocity may be written in inertial coordinates using Eqns.~\eqref{e:relativeVelocityAlphadotQdotForm} and \eqref{e:Inertial_to_Gamma} as
\begin{equation}\label{e:deltaV_in_Inertial}
\begin{split}
\vec{v}_j-&\vec{v}_i=\\
&\left(\mp(1+\cot^2\alpha_{ij})\cos\phi_{ij}L\dot{\alpha}_{ij}-r_{ij}(\dot{Q}_{ij}+\dot{\theta}_i)\sin\phi_{ij}\right) \hat{e}_x+\\
&\left(\mp(1+\cot^2\alpha_{ij})\sin\phi_{ij}L\dot{\alpha}_{ij}+r_{ij}(\dot{Q}_{ij}+\dot{\theta}_i)\cos\phi_{ij}\right) \hat{e}_y.\hspace{.4em}\end{split} \end{equation}

Equation \eqref{e:deltaV_in_Inertial} can now be used to quantify the relative velocities of the agents in inertial frame from optic flow $\dot{Q}_{ij}$ and angular expansion rate $\dot{\alpha}_{ij}$.

\paragraph{Idealized speed/orientation update rules} The desired value to update the velocity of agent $i$ is given by C-S model as 
\[\dot{\vec{v}}_i^* =\frac{H(\vec{v}_j-\vec{v}_i)}{(1+r^2_{ij})^{\beta}},\]
which using Eqn.~\eqref{e:rij-LcotAlpha} can be written as
\begin{equation}\label{e:dot_Vstar_1}
\dot{\vec{v}}_i^*=\frac{H(\vec{v}_j-\vec{v}_i)}{(1+L^2 \cot^2\alpha_{ij})^{\beta}}.
\end{equation}
$\dot{\vec{v}}^*_i$ can be written in inertial frame components as
\begin{equation}\label{e:dot_Vstar_2}
\dot{\vec{v}}_i^*=\dot{v}_x^{i*} \hat{e}_x+\dot{v}_y^{i*} \hat{e}_y.
\end{equation}
From Eqns.~\eqref{e:deltaV_in_Inertial}, \eqref{e:dot_Vstar_1}, and \eqref{e:dot_Vstar_2}, the components are drawn as
\begin{equation}\label{e:dot_vStar_x}
\begin{split}
\dot{v}_x^{i*}&\!=\!\\&H\frac{\mp(1+\cot^2\alpha_{ij})\cos(\phi_{ij})L\dot{\alpha}_{ij}-r_{ij}(\dot{Q}_{ij}+\dot{\theta}_i)\sin(\phi_{ij})}{(1+L^2 \cot^2\alpha_{ij})^{\beta}}\end{split}\end{equation}
and
\begin{equation}\label{e:dot_vStar_y}\begin{split}
\dot{v}_y^{i*}&\!=\!\\&H\frac{ \mp(1+\cot^2\alpha_{ij})\sin(\phi_{ij})L\dot{\alpha}_{ij}+r_{ij}(\dot{Q}_{ij}+\dot{\theta}_i)\cos(\phi_{ij})}{(1+L^2 \cot^2\alpha_{ij})^{\beta}}.
\end{split}\end{equation}
One can now use these two components together with Eqns.~\eqref{e:V_dot} and~\eqref{e:Theta_dot} to find the desired time derivatives of the agent $i$'s speed  and orientation as
\begin{equation}\label{e:dot_vStar_3}
\dot{v}_i^*=\frac{v_x^i\dot{v}_x^{i*}+v^i_{y}\dot{v}_y^{i*}}{v_i}\end{equation}
and 
\begin{equation}\label{e:dot_Theta_Star_3}
\dot{\theta}_i^*=\frac{1}{1+\tan^2(\theta_i)}\frac{v^i_{x}\dot{v}^i_{y}-v^i_{y}\dot{v}^i_{x}}{v^2_x}.
\end{equation}

\noindent By substitution of $v^i_x=v_i \cos(\theta_i)$, $v^i_y=v_i \sin(\theta_i)$, and $\dot{v}_x^{i*}$, $\dot{v}_y^{i*}$ from Eqn.~\eqref{e:dot_vStar_x} and \eqref{e:dot_vStar_y} in Eqns.~\eqref{e:dot_vStar_3} and~\eqref{e:dot_Theta_Star_3}, we obtain the desired values for time derivatives of speed and orientation of agent $i$ as
\begin{equation}\label{e:dot_vStar_4}\footnotesize
\begin{split}\hspace{-.25em}
\dot{v}_i^*&\!=\! H\cos\theta_i\!\left( \frac{\mp (1\!+\!\cot^2{\alpha_{ij})L\dot{\alpha}}\cos\phi_{ij}-L\cot\alpha_{ij}(\dot{Q}_{ij}\!+\!\dot{\theta_i})\sin\phi_{ij}}{(1+L^2\cot^2\alpha_{ij})^{\beta}}\right)\\
&+H\sin\theta_i\!\left( \frac{\mp(1\!+\!\cot^2{\alpha_{ij})L\dot{\alpha}}\sin\phi_{ij}+L\cot\alpha_{ij}(\dot{Q}_{ij}\!+\!\dot{\theta_i})\cos\phi_{ij}}{(1+L^2\cot^2\alpha_{ij})^{\beta}}\right)\end{split}\end{equation}

\begin{equation}\label{e:dot_Theta_Star_4}\small
\begin{split}\hspace{-.4em}
\dot{\theta}_i^*&\!=\! \frac{-H\sin\theta_i}{v_i}\!\left( \frac{\mp(1\!+\!\cot^2{\alpha_{ij})L\dot{\alpha}}\cos\phi_{ij}\!-\!L\cot\alpha_{ij}(\dot{Q}_{ij}\!+\!\dot{\theta_i})\sin\phi_{ij}}{(1+L^2\cot^2\alpha_{ij})^{\beta}}\right)\\
&\!+\!\frac{H\cos\theta_i}{v_i}\!\left( \frac{\mp(1\!+\!\cot^2{\alpha_{ij})L\dot{\alpha}}\sin\phi_{ij}+L\cot\alpha_{ij}(\dot{Q}_{ij}\!+\!\dot{\theta_i})\cos\phi_{ij}}{(1+L^2\cot^2\alpha_{ij})^{\beta}}\right).\end{split}\end{equation}
 When agent $i$ responds to a collection of $N$ agents (agent $i$ updates its speed and orientation using sum of the values obtained from Eqns.~\eqref{e:dot_vStar_4} and~\eqref{e:dot_Theta_Star_4} for $j=1,..,N$), then the speed and orientation update laws become
 
\begin{equation}\small\begin{split}
\dot{v}_i^*\!=\! H\cos\theta_i\sum_{j=1}^N\frac{\mp (1\!+\!\cot^2{\alpha_{ij})L\dot{\alpha}}\cos\phi_{ij}-L\cot\alpha_{ij}(\dot{Q}_{ij}\!+\!\dot{\theta_i})\sin\phi_{ij}}{(1+L^2\cot^2\alpha_{ij})^{\beta}}\\
+H\sin\theta_i\sum_{j=1}^N \frac{\mp(1\!+\!\cot^2{\alpha_{ij})L\dot{\alpha}}\sin\phi_{ij}+L\cot\alpha_{ij}(\dot{Q}_{ij}\!+\!\dot{\theta_i})\cos\phi_{ij}}{(1+L^2\cot^2\alpha_{ij})^{\beta}}
\end{split} \end{equation}
and
\begin{equation}\small\begin{split}\dot{\theta}_i^*&=\\ &\frac{-H\sin\theta_i}{v_i}\sum_{j=1}^N \frac{\mp(1\!+\!\cot^2{\alpha_{ij})L\dot{\alpha}}\cos\phi_{ij}-L\cot\alpha_{ij}(\dot{Q}_{ij}\!+\!\dot{\theta_i})\sin\phi_{ij}}{(1+L^2\cot^2\alpha_{ij})^{\beta}}\\
&+\!\frac{H\cos\theta_i}{v_i}\sum_{j=1}^N \frac{\mp(1\!+\!\cot^2{\alpha_{ij})L\dot{\alpha}}\sin\phi_{ij}+L\cot\alpha_{ij}(\dot{Q}_{ij}\!+\!\dot{\theta_i})\cos\phi_{ij}}{(1+L^2\cot^2\alpha_{ij})^{\beta}}.\end{split} \end{equation}

\noindent These equations can be simplified as
\begin{equation}\label{e:dot_vStar_5}
\dot{v}_i^*\!=\! HL\sum_{j=1}^N \frac{\mp\dot{\alpha}_{ij}(1\!+\!\cot^2{\alpha_{ij})}\cos\gamma_{ij}-(\dot{Q}_{ij}\!+\!\dot{\theta_i})\cot\alpha_{ij}\sin\gamma_{ij}}{(1+L^2\cot^2\alpha_{ij})^{\beta}}
\end{equation}

\begin{equation}\label{e:dot_Theta_Star_5}
\dot{\theta}^*_i\!=\! \frac{HL}{v_i}\sum_{j=1}^N\frac{\mp\dot{\alpha}_{ij}(1\!+\!\cot^2{\alpha_{ij})}\sin\gamma_{ij}+(\dot{Q}_{ij}\!+\!\dot{\theta}_i)\cot\alpha_{ij}\cos\gamma_{ij}}{(1+L^2\cot^2\alpha_{ij})^{\beta}}.\end{equation}

\paragraph{Final update rules for speed and orientation} Now we can use the desired value of $\dot{v}_i^*$ as in Eqn.~\eqref{e:dot_vStar_5} for speed control of agent $i$ in the form of $\dot{v}_i=\dot{v}_i^*$ or
\[\dot{v}_i\!=\! HL\sum_{j=1}^N \frac{\mp\dot{\alpha}_{ij}(1\!+\!\cot^2{\alpha_{ij})}\cos\gamma_{ij}-(\dot{Q}_{ij}\!+\!\dot{\theta_i})\cot\alpha_{ij}\sin\gamma_{ij}}{(1+L^2\cot^2\alpha_{ij})^{\beta}}\]
In order to regulate $\dot{\theta}_i$ to its desired value $\dot{\theta}_i^*$ given by Eqn.~\eqref{e:dot_Theta_Star_5}, we design a $\ddot{\theta}_i$ controller to be $u_{\omega}=-k(\dot{\theta}_i-\dot{\theta}^*_i)$ or
\begin{equation*}\begin{split}\hspace{-.25em} \ddot{\theta}_i&\!=\!-k\dot{\theta}_i\\&\hspace{-0.2em}+\!\tiny\frac{kHL}{v_i}\!\sum_{j=1}^N\frac{\normalsize\mp\dot{\alpha}_{ij}(1\!+\!\cot^2{\alpha_{ij}})\sin\gamma_{ij}\!+\!(\dot{Q}_{ij}\!+\!\dot{\theta}_i)\cot\alpha_{ij}\cos\gamma_{ij}}{(1+L^2\cot^2\alpha_{ij})^{\beta}},\end{split}\end{equation*}\normalsize
where $k>0$ denotes the feedback gain.\end{proof}

\subsection{Agent size $(L)$}\label{ss:agentSize}
The appearance of half-size $L$ in the feedback raises a discussion of agent size. In particular, does convergence require agents to know $L$?

Agents do not need to know the length $L$. The convergence proof holds when $L_e$ is replaced by any real positive constant $L_e>0$, as proven in the following corollary to Thm.~\ref{t:OfFlockingThm}.
\begin{corollary} In the case that $L$ is an unknown constant, replace it with an arbitrary real positive number $L_e>0$. Then Thm.~\ref{t:OfFlockingThm} holds.\end{corollary}
\begin{proof}
The C-S dynamics in Eq.~\eqref{e:Continuous_CS_Dynamics} can be written in the form of
\begin{equation}\label{e:proof1}
\dot{\vec{v}}_i=\sum_{j=1}^N \frac{H\sigma^{-2\beta}(\vec{v_j}-\vec{v_i})}{\left[1+\sigma^{-2\beta}r_{ij} ^2\right]^{\beta}},
\end{equation}
We can multiply the numerator in Eq.~\eqref{e:proof1} by a positive number $L/L_e$ without affecting the convergence of the C-S dynamics. Given that convergence holds for any $\sigma>0$, then we can also choose $\sigma$ such that $\sigma^{-2\beta}=\frac{L_e^2}{L^2}$. Then Eq.~\eqref{e:proof1} becomes
\begin{equation}\label{e:proof2}
\dot{\vec{v}}_i=\sum_{j=1}^N \frac{\frac{L_e}{L} H(\vec{v_j}-\vec{v_i})}{\left[1+\frac{L_e^2}{L^2}r_{ij} ^2\right]^{\beta}},
\end{equation}
Choosing Eq.~\eqref{e:proof2} to be the desired value of the velocity derivative in Eq.~\eqref{e:dot_Vstar_1} as
\[\dot{\vec{v}}_i^*=\sum_{j=1}^N \frac{\frac{L_e}{L} H(\vec{v_j}-\vec{v_i})}{\left[1+\frac{L_e^2}{L^2}r_{ij} ^2\right]^{\beta}},\]
and following the same procedure, we reach
\begin{equation} \hspace{-.3em}
\dot{v}_i\!=\!HL_e\sum_{j=1}^N \frac{\mp\dot{\alpha}_{ij}(1\!+\!\cot^2{\alpha_{ij})}\cos\gamma_{ij}-\left(\dot{Q}_{ij}\!+\!\dot{\theta_i}\right)\cot\alpha_{ij}\sin\gamma_{ij}}{\left(1+L_e^2\cot^2\alpha_{ij}\right)^{\beta}}\label{e:vdotLe}
\end{equation}
and 
\begin{equation}\begin{split}
\ddot{\theta}_i=-k\dot{\theta}_i+\hspace{.4\textwidth}\\
\frac{kHL_e}{v_i}\sum_{j=1}^N\frac{\mp\dot{\alpha}_{ij}(1\!+\!\cot^2{\alpha_{ij})}\sin\gamma_{ij}+\left(\dot{Q}_{ij}\!+\!\dot{\theta}_i\right)\cot\alpha_{ij}\cos\gamma_{ij}}{\left(1+L_e^2\cot^2\alpha_{ij}\right)^{\beta}}.\end{split}\label{e:thetaDotLe}\end{equation}\end{proof}

Thus, the size knowledge question does not impose a reasonable challenge to convergence under YFM feedback.

\section{Numerical simulation} 
Simulations were conducted in MATLAB implementing Eqns. ~\eqref{e:Agent_Dynamics},~\eqref{e:V_dot}-\eqref{e:qdot}, and incorporating noisy measurements and perception limits, and $\gamma_{ij}$ computed by $ \gamma_{ij}=\arctan{\frac{y_j-y_i}{x_j-x_i}}-\theta_i$. Additive measurement noise was included on $\dot{Q}$ and $\alpha$ as
\begin{align*}\dot{Q}_n&=\dot{Q}+\sigma_q\mathtt{randn}\\
\alpha_n&=\alpha+\sigma_a\mathtt{randn},\end{align*}
where $\sigma_q$ and $\sigma_a$ denote the standard deviation of noise on $\alpha$ and $\sigma_q$, respectively, and $\mathtt{randn}$ denotes a zero mean unit intensity normal distribution. Simulations used Euler integration  and a time step of 0.01 seconds, and began from the initial positions chosen randomly as shown in Fig.~\ref{f:2dTrajNoisy} (circles). Simulations involving noise incorporate a minimum visibility limit $\munderbar{\alpha}$, which assumes the sensing and feedback path cannot resolve agents smaller than this angular size. Except where specified otherwise, simulations used the parameters in Table \ref{t:simulationParameters}.
\begin{table}[h] \centering
\caption{Simulation parameters: constants $\beta,H$, feedback gain $k$, agent size $L$, and $\munderbar{\alpha}$ visibility limit. }
\begin{tabular}{ccccc} \toprule
\multicolumn{3}{c}{Feedback} & \multicolumn{2}{c}{Simulation}\\
\cmidrule(lr){1-3} \cmidrule(lr){4-5}
$\beta$ & $H$& $k$ & $L$ (m) & $\munderbar{\alpha}$ (rad) \\
\cmidrule{1-5}
\centering 0.4 &  1 &    20 &   1 & 0.005 \\ \bottomrule
\end{tabular}\label{t:simulationParameters}\end{table}
For simulations exhibiting decaying oscillation, an equivalent natural frequency and damping ratio was computed from an agent time history by applying the logarithmic decrement method to successive peaks \citep{palm2010systemDynamics}.

\subsection{Visually-guided YFM and idealized C-S comparison}
In this simulation, $\beta=0.4$ is chosen and five agents with different initial speeds and orientations are used. Figure \ref{f:velocities} and \ref{f:orientations} show that both speeds and orientations converge to a common value. The visually-driven convergence is slower than observed in the perfect information Cucker-Smale case, and also includes oscillations not seen in the C-S case. Although the convergence to flocking definition is provided, the wider convergence envelope and oscillation structure gives rise to a more chaotic looking trajectory structure.

\begin{figure}[htbp]\centering
\begin{subfigure}{0.24\textwidth}
    \includegraphics[width=\textwidth]{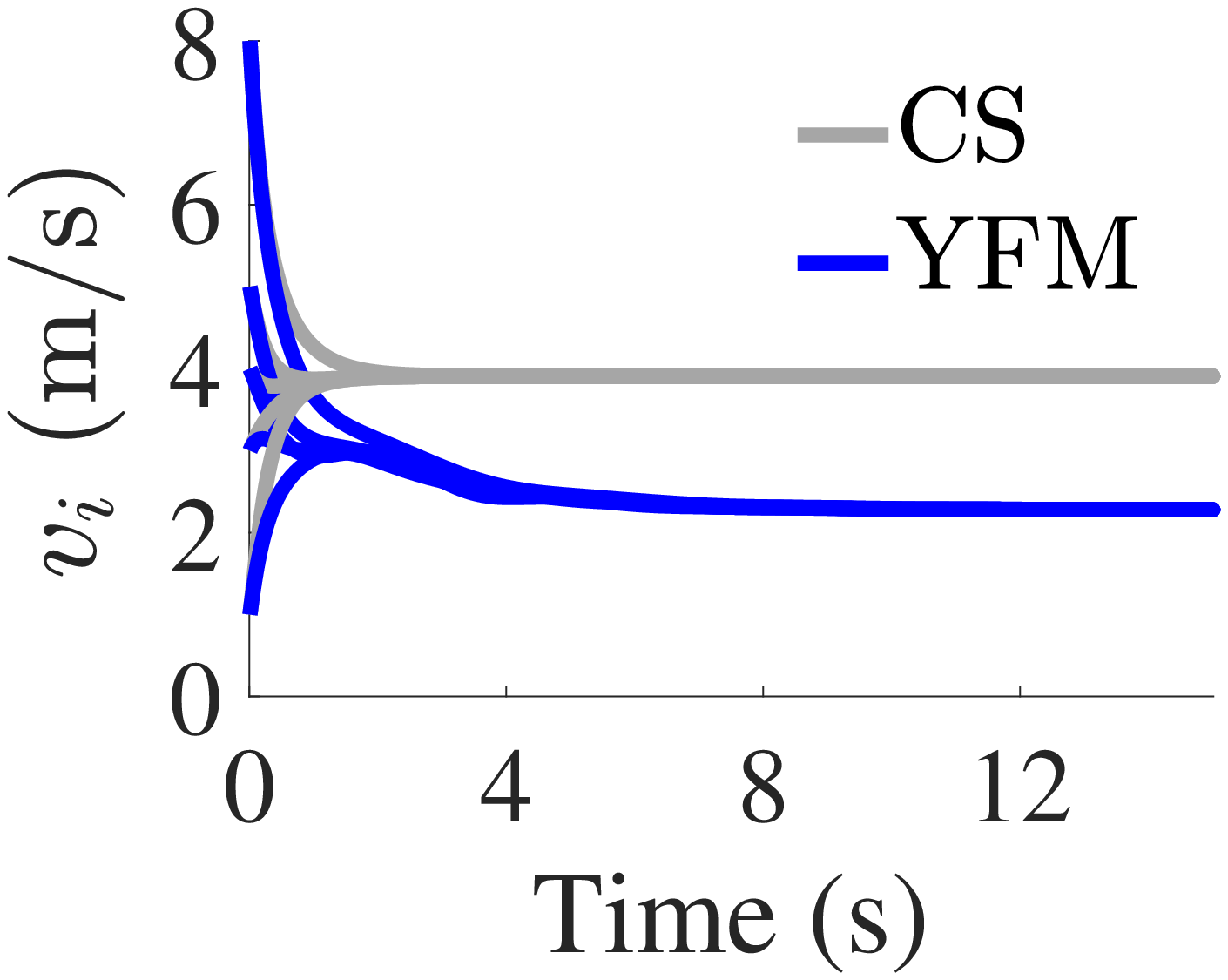}
    \caption{Agent speed convergence}
    \label{f:velocities}
\end{subfigure}
\begin{subfigure}{0.24\textwidth}
\includegraphics[width=\textwidth]{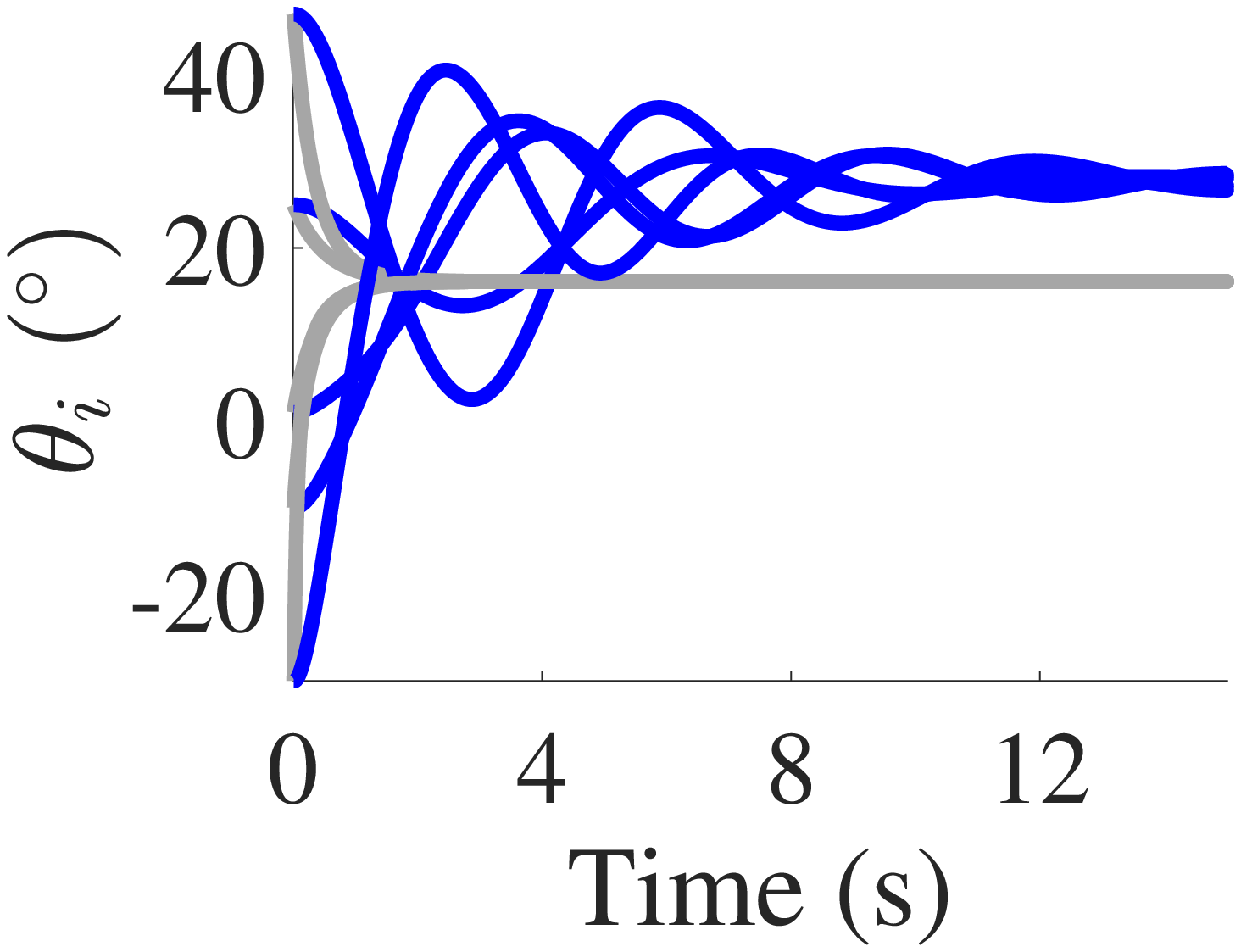}     
\caption{Agent orientation convergence.}                          \label{f:orientations}
\end{subfigure}
\caption{YFM agent's velocities (speed and orientation), showing convergence under low gain ($k=0.2$) and comparison to idealized Cucker-Smale flocking.}\end{figure}

\subsection{Parameter sensitivity}
Parameter sensitivity was assessed in simulation for $H,k$, and $L$, as seen in Fig.~\ref{f:simVary}. The number of oscillations increases with increasing $H$ (Fig.~\ref{f:simVaryH}) and with decreasing $k$ (Fig.~\ref{f:simVaryk}). Functionally, the $H$ parameter appears in both velocity and orientation update equations, while the orientation gain $k$ enters only in the orientation update equation. As predicted by Section \ref{ss:agentSize}, Fig.~\ref{f:simVaryL} shows that $L$ has no effect on convergence rate or behavior.

\begin{figure}[htbp]\centering
\begin{subfigure}{0.2521\textwidth}
    \includegraphics[trim=6 0 20 0, clip,width=\textwidth]{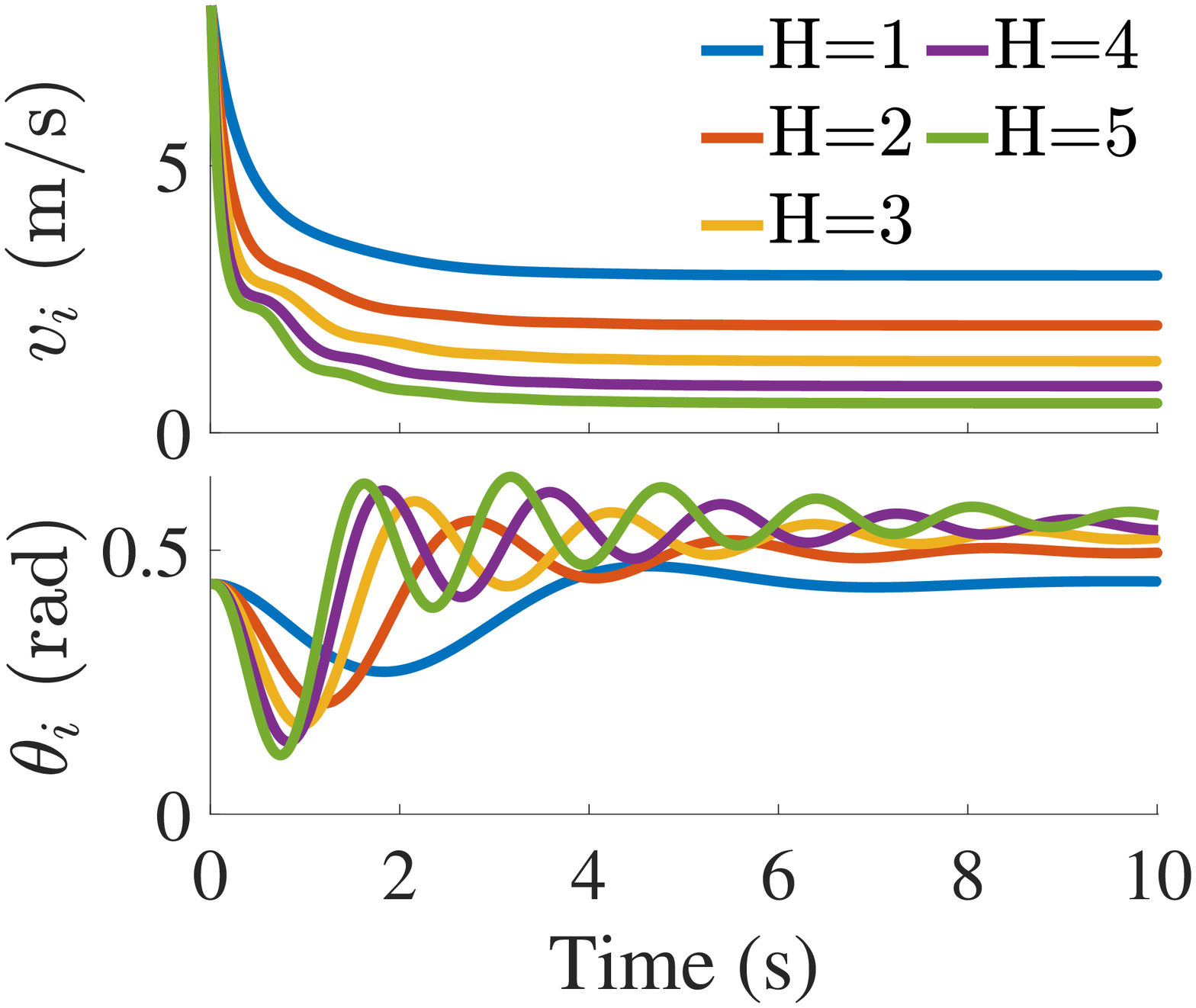}\caption{Varying $H$}\label{f:simVaryH}\end{subfigure}
\begin{subfigure}{.2521\textwidth}
	\includegraphics[trim=6 0 20 0, clip, width=\textwidth]{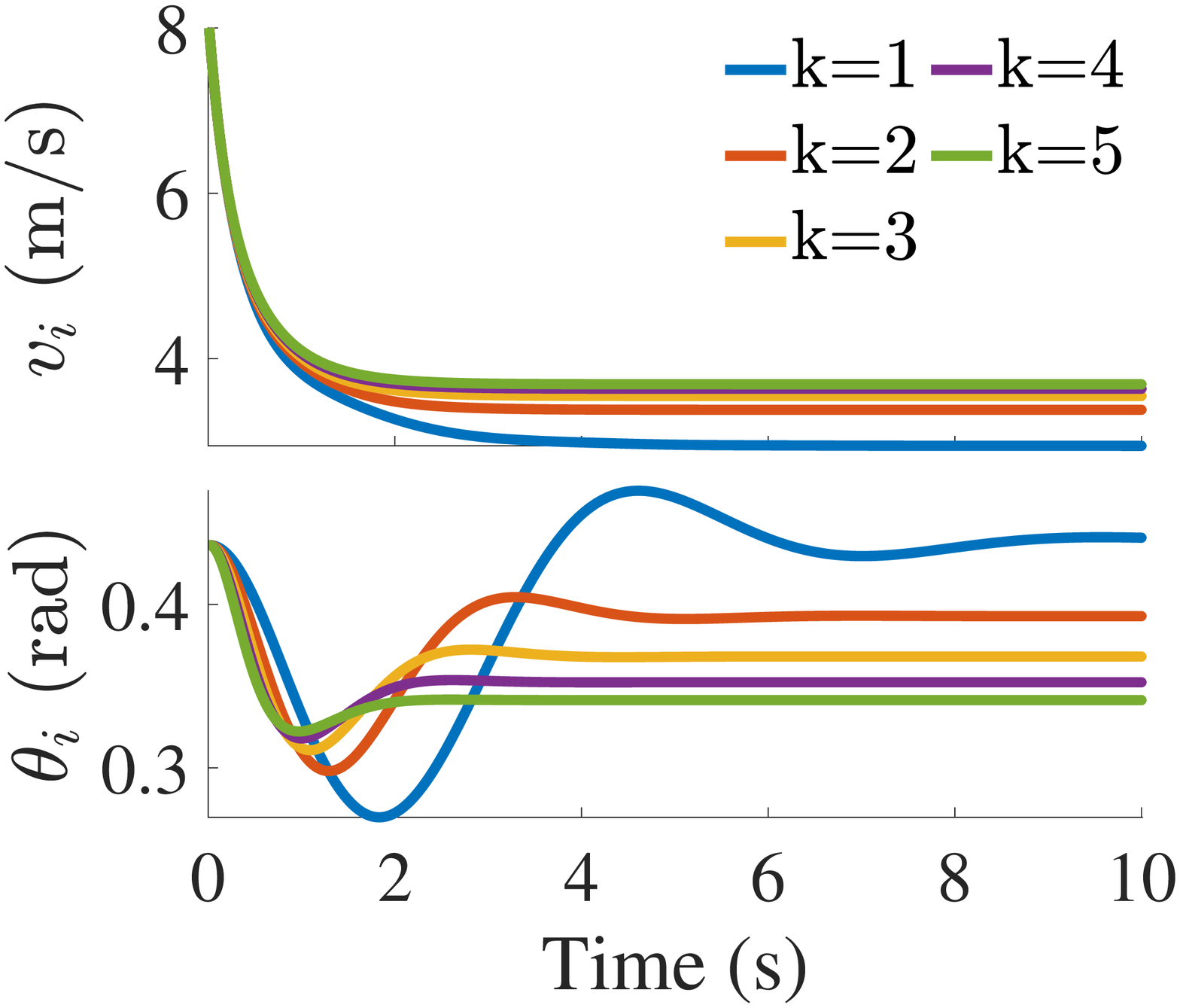} \caption{Varying $k$}\label{f:simVaryk}\end{subfigure}
\begin{subfigure}{0.25\textwidth}
\includegraphics[width=\textwidth]{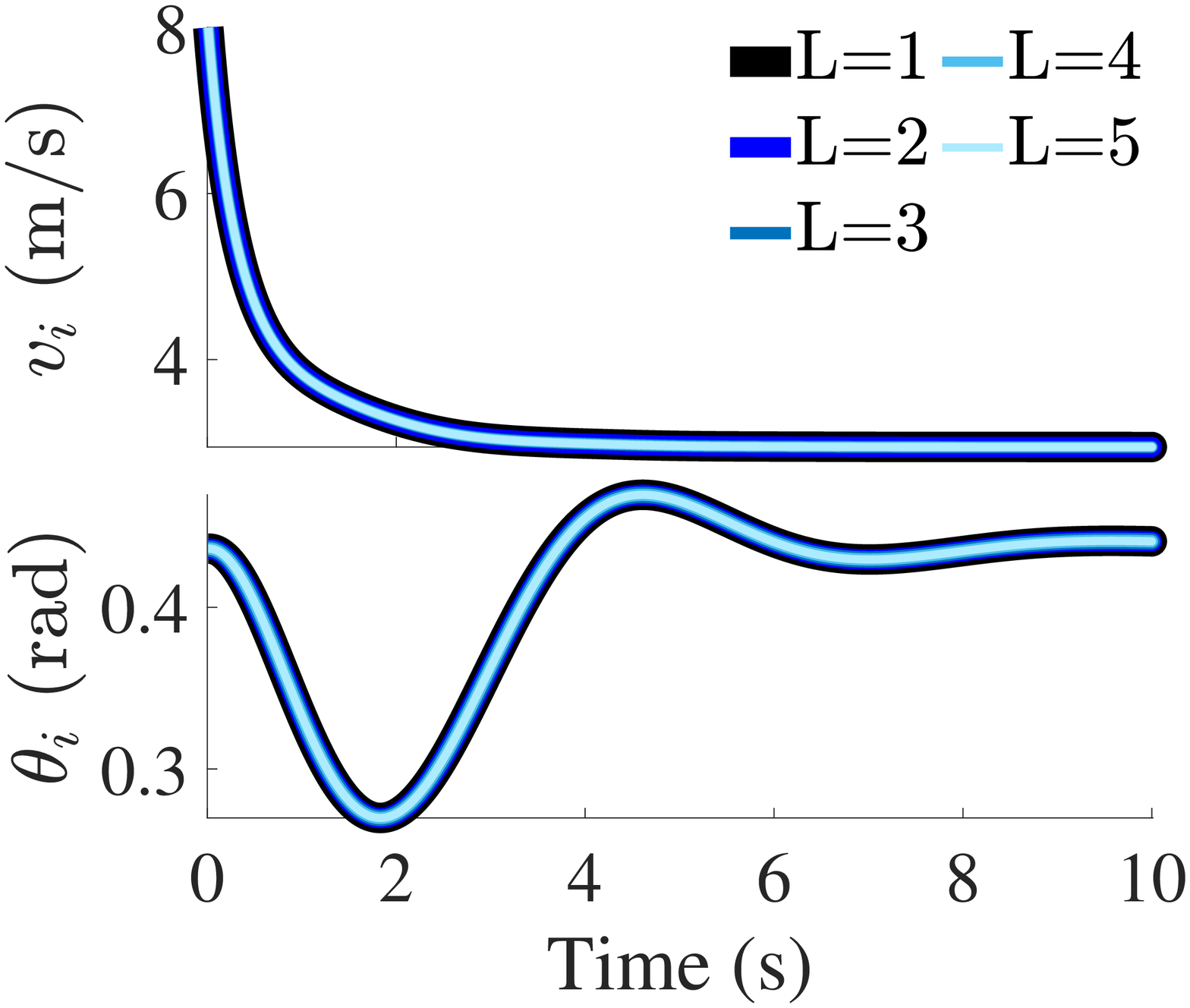}\caption{Varying $L$ (meters).}\label{f:simVaryL}\end{subfigure}
\caption{A single agent in the swarm is illustrated during parameter variation, showing that $k$ and $H$ variation affects convergence rate, oscillation period, and converged asymptote, while results are unchanged under $L$ variation. (During parameter variation, the remaining two parameters were set to 1.)}\label{f:simVary}\end{figure}

\subsection{Linearity assessment}

Although the simulations show decaying oscillations for low values of gain $k$ or $H$, one must not assume such behavior can be described by a linear system solution having a sinusoid of constant frequency bounded by a decaying exponential envelope. In particular, Fig.~\ref{f:equivDampingNatlFreq} shows that the oscillations of a single agent have a period growing with increasing time (e.g., the natural frequency decreases), and the equivalent damping ratio also decreases with increasing time. 
\begin{figure}\centering
    \begin{subfigure}{0.23\textwidth}\includegraphics[width=\textwidth,trim=0 50 0 0, clip]{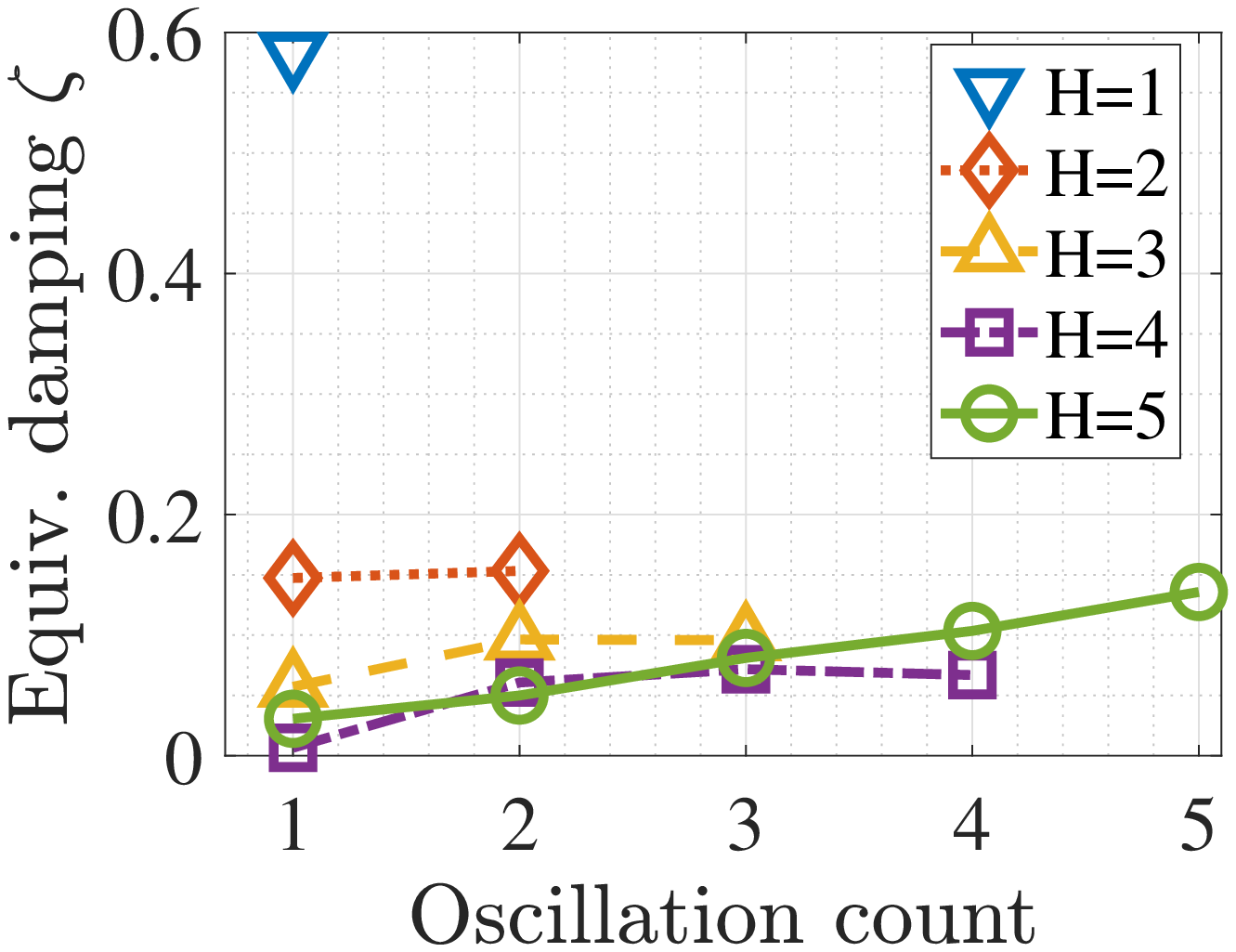}
    \end{subfigure}
    \hspace{4pt}
    \begin{subfigure}{0.23\textwidth}\includegraphics[width=\textwidth,,trim=0 50 0 0, clip]{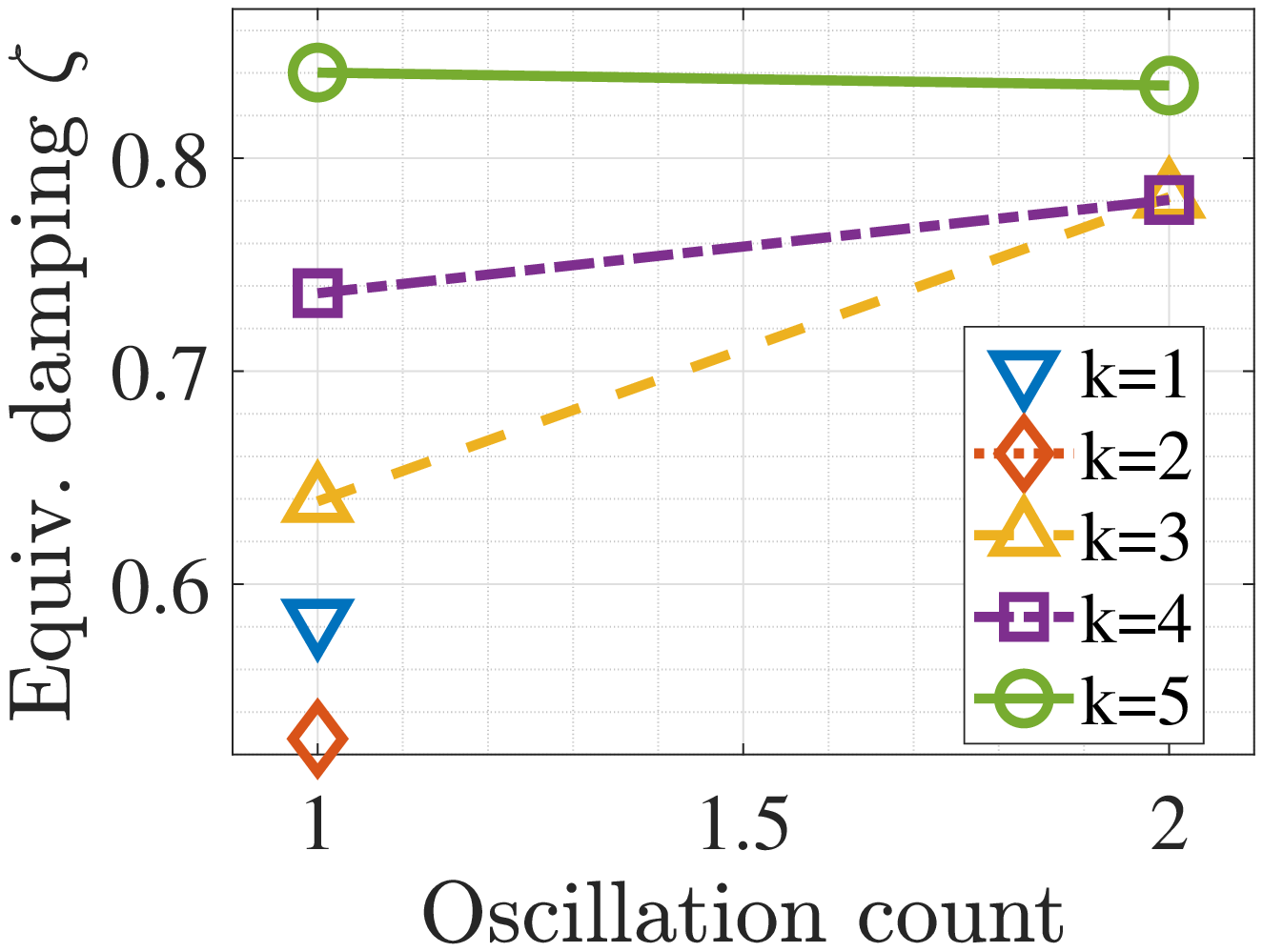}
    \end{subfigure}
\\ \vspace{5pt}
    \begin{subfigure}{0.23\textwidth}\includegraphics[width=\textwidth]{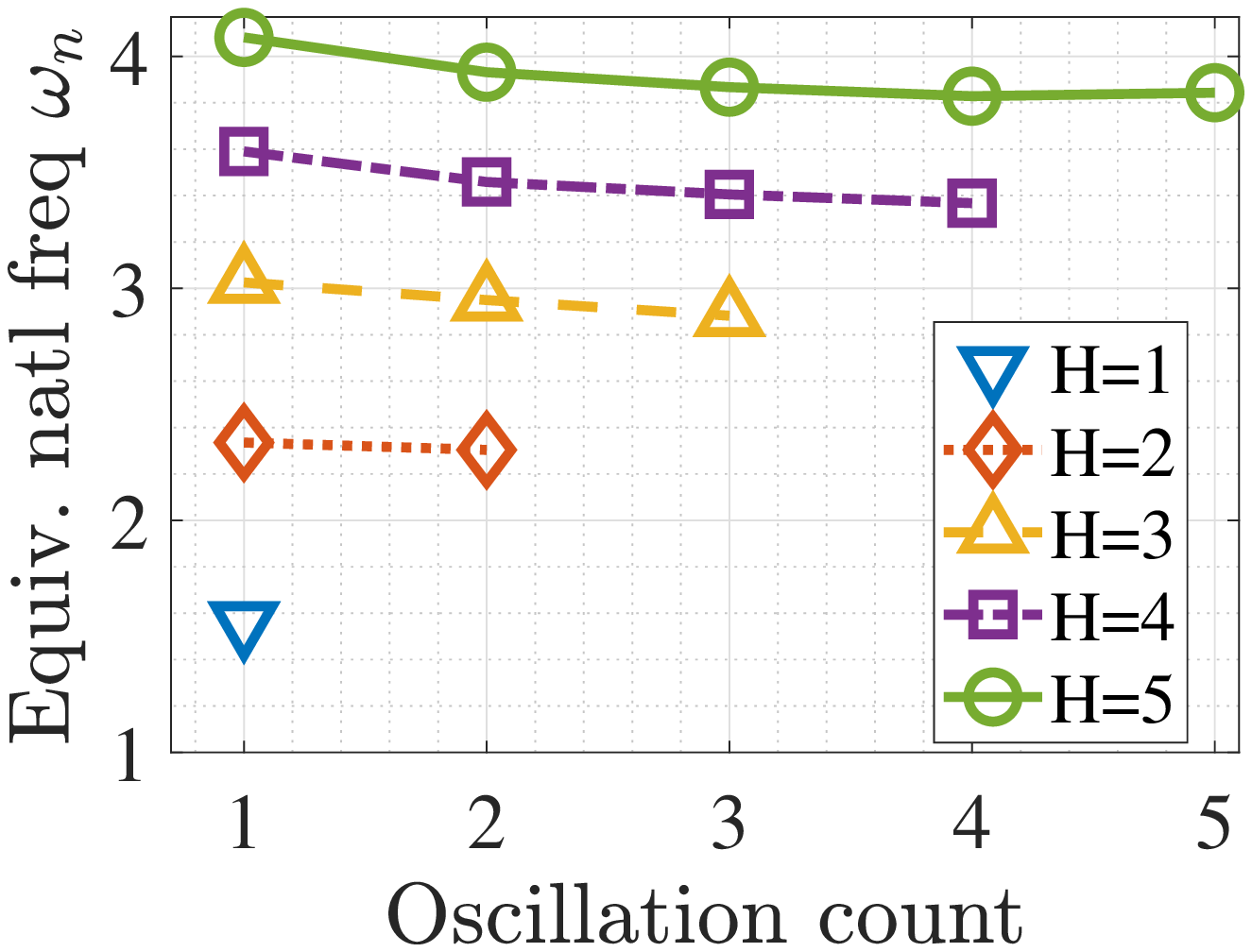}
    \caption{Coupling constant $H$} \label{f:oscillationH}
    \end{subfigure}\hspace{4pt}
    \begin{subfigure}{0.23\textwidth}\includegraphics[width=\textwidth]{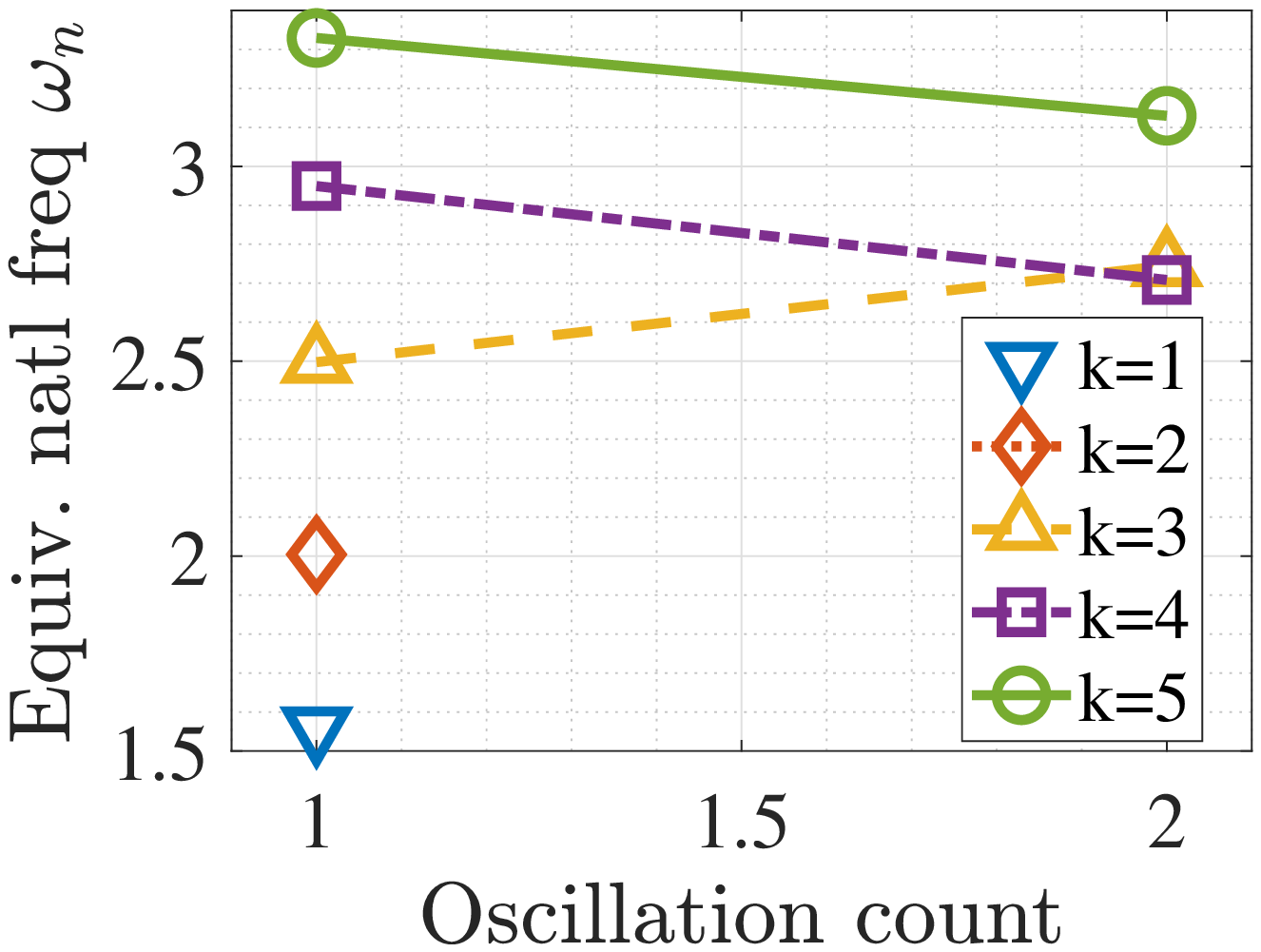}
    \caption{Gain $k$} \label{fig:oscillationk}\end{subfigure}
    \caption{The equivalent damping ratios and natural frequencies, as quantified by successive peaks in the first agent's heading angle history, show time varying frequencies and damping and suggest the behavior is inherently nonlinear. }\label{f:equivDampingNatlFreq}
\end{figure}

\subsection{Effects of measurement noise}
Trajectory plots, as seen in Fig.~\ref{f:2dTrajNoisy}, illustrate that the asymptote may be affected by the inclusion of measurement noise.
\begin{figure}\centering
    \includegraphics[width=0.33\textwidth]{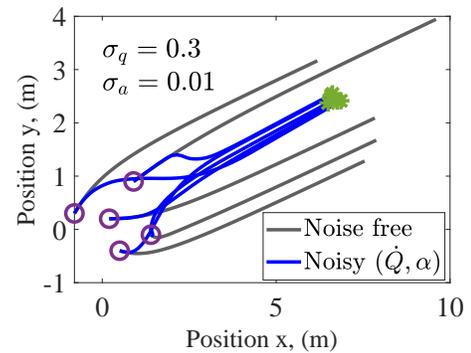}
    \caption{{Agent position trajectories, showing performance with and without additive measurement noise. Circles/stars indicate initial/terminal agent positions. ($L$=.08m, $T_\mathrm{max}$=2sec).}}\label{f:2dTrajNoisy}\end{figure}
Additional simulations incorporating noise, as illustrated in Fig.~\ref{f:varyingMeasurementNoise}, illustrate that the system behavior does show degradation in performance with increasing noise level. Convergence, as measured by monotonic convergence to an asymptote, is more sensitive to noise on $\dot{\alpha}$ than on $\dot{Q}$.  

\begin{figure}[htbp]\begin{center}
\includegraphics[width=0.24\textwidth]{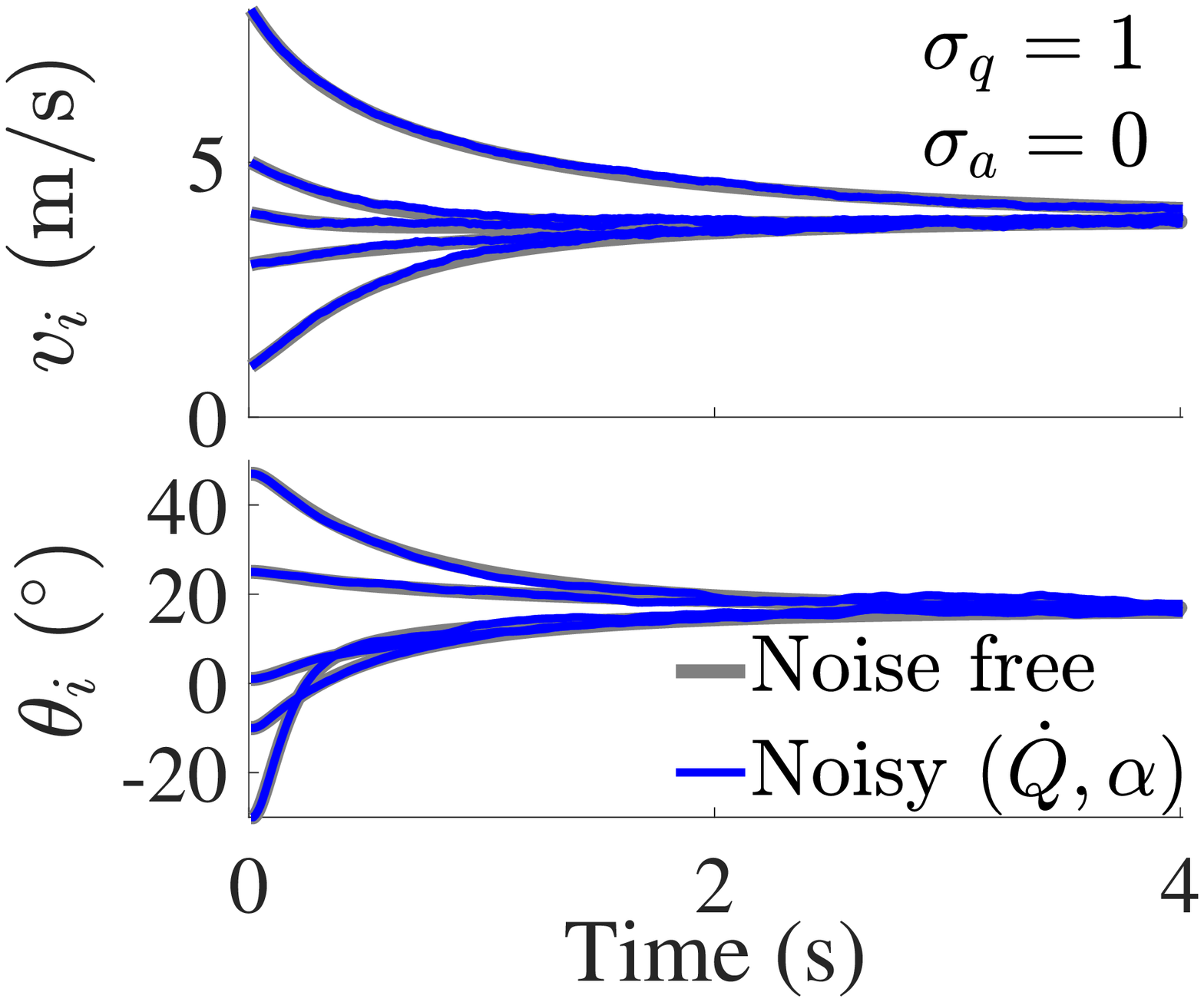}
\includegraphics[width=0.24\textwidth]{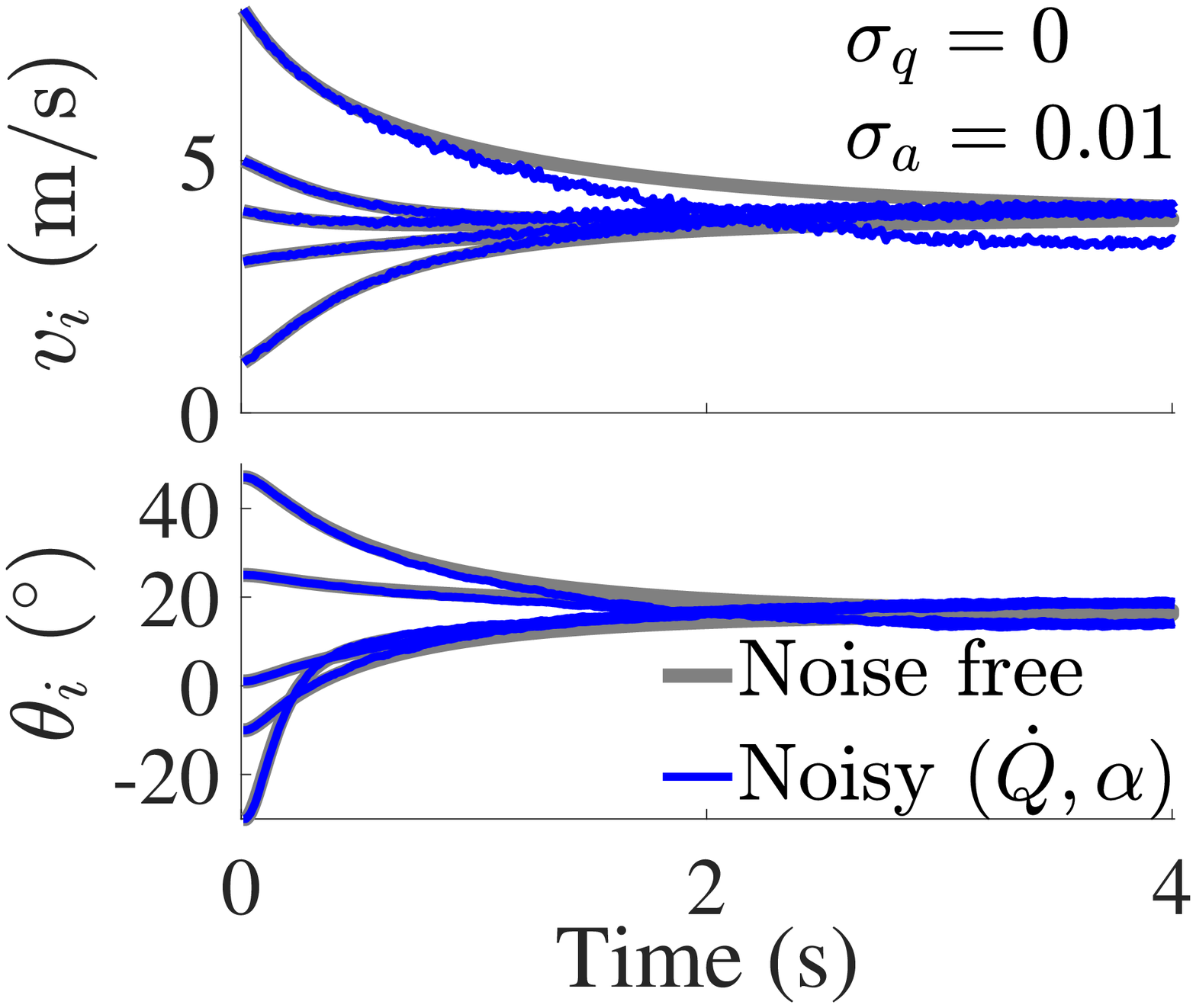}
\includegraphics[width=0.24\textwidth]{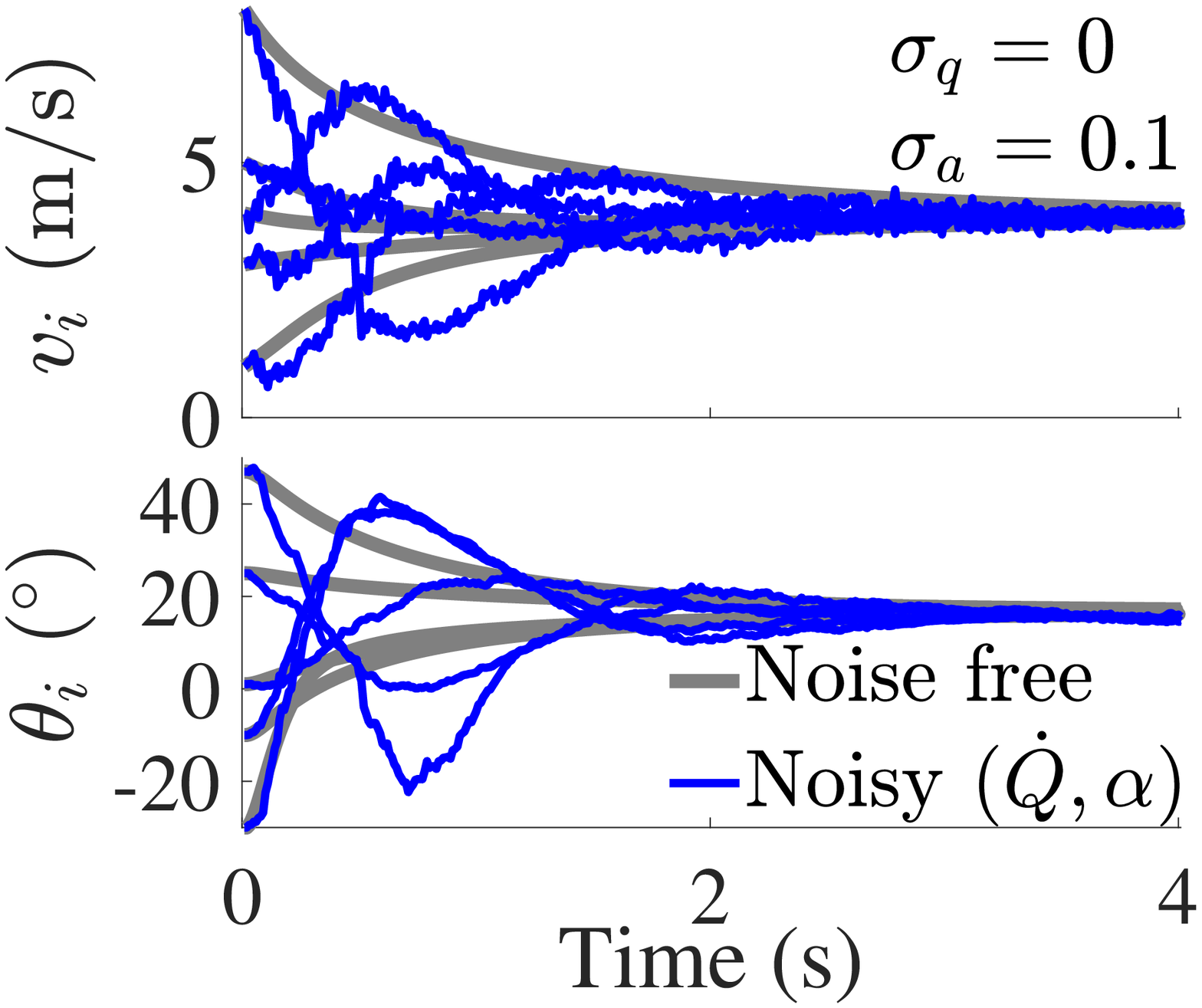}
\includegraphics[width=0.24\textwidth]{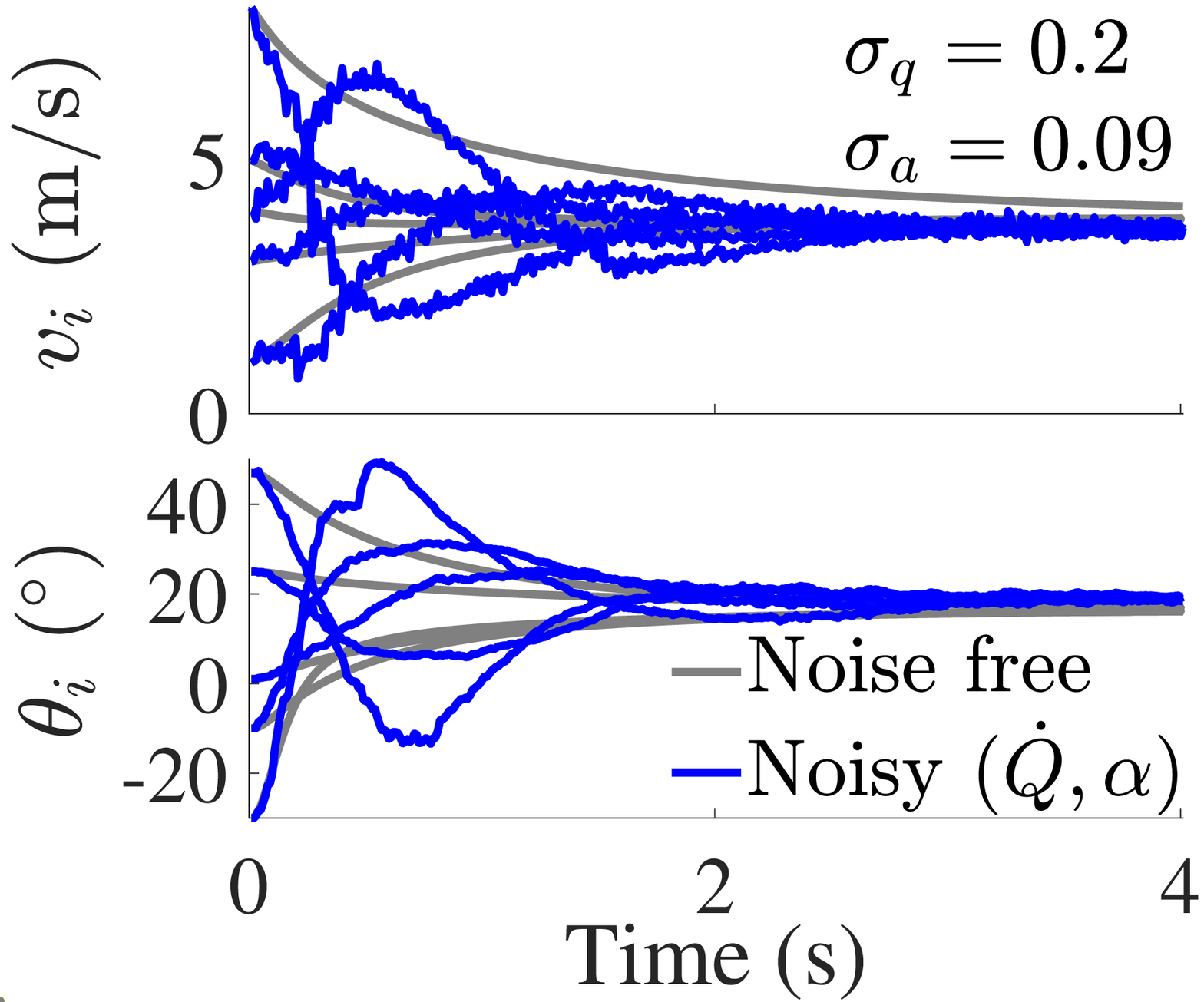}
\caption{YFM convergence under varying levels of measurement noise on $\dot{Q}$ and $\alpha$.}\label{f:varyingMeasurementNoise}\end{center}\end{figure}

\section{Discussion}\label{Discussion}

\paragraph{Measurement noise}
Theorem \ref{t:OfFlockingThm} describes convergence in the idealized sensor (noise-free) case. Realistic implementations involve noise, which has been previously shown to be capable of both supporting and discouraging C-S flocking \citep{Sun2015csNoiseMultPositive}. 
We can derive at least one criteria for noise-induced flocking if following assumptions are true:
\begin{description}
\item[A1] Agent $i$ ignores visual signals in the regions  $\gamma \in (-\Gamma,+\Gamma)$ and $\gamma \in (\pi-\Gamma,\pi+\Gamma)$ as seen in Fig.~\ref{f:blindSpots}, (e.g., the field of view of the agents are limited such that there are blind spot located directly in front of and behind the agents~\citep{soria2019LatVisionFlocking}),

and
\item[A2]  There is an upper bound $\rho$ for the size of the group, in other words $t\geq 0$, $\forall i,j=1,...,N:$ $r_{ij}\leq \rho$. Such a restriction may have an origin in sensing limitations.\end{description}

\begin{figure}[htbp]\centering
    \includegraphics[width=0.33\textwidth]{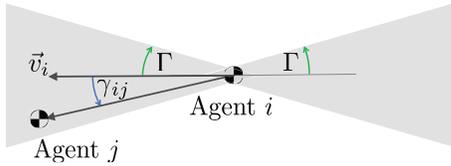}
    \caption{The occlusal angle $\Gamma$ used in assumption A1 specifies the region (gray) where agent $i$ ignores optic flow signals from neighboring agents ($j$)    }\label{f:blindSpots}\end{figure}

Under these assumptions, a bound for the noise on the optic flow signal may be found.

\begin{prop}
Suppose the C-S dynamics in a discretized form of Eqn. \eqref{e:Continuous_CS_Dynamics} are convergent for bounded velocity noise
\[v_{j,n} = v_j+n, \mathrm{~s.t.~} |n|<\bar{n},\]
$\alpha$ is measured with no noise, and that optic flow measurements contain additive noise $q_{ij}$ as \[\dot{Q}_{ij,n} = \dot{Q}_ij + q_{ij},\] and sensory assumptions A1 and A2 apply. Then, the YFM dynamics in Eqns.~\eqref{e:vdot}-\eqref{e:thetaDot} are convergent under bounded optic flow noise $|q_{ij}|<\bar{q},$ where
\begin{equation}\bar{q}=\frac{\bar{n}\sin{\Gamma}}{\rho}. \label{e:qnoiseBoundGamma}\end{equation}\end{prop}

\begin{proof}Let $n$ be a random variable added to measurements of agent $j$'s speed, such that optic flow is now 

\begin{equation}\dot{Q}_{ij}(v_j+n_{ij})=-\dot{\theta}_i+\frac{1}{r_{ij}}(v_i \sin\gamma_{ij}-(v_j+n_{ij}) \sin \gamma_{ji}).\end{equation}

Denote $q$ as the resulting change in optic flow from the noise-free case. Defining a nearness function $\mu_{ij}=1/r_{ij}$ and applying Eqn.~\eqref{e:qdot},
\begin{equation}
q_{ij}=\dot{Q}(v_j+n)-\dot{Q}(v_j)=-\mu_{ij} n_{ij} \sin {\gamma}_{ij}.\end{equation}
When a convergence guarantee for additive velocity noise with magnitude bounded by $\bar{n}$ is available, then YFM is convergent for a noise of $q_{ij}\leq |-\mu_{ij} \bar{n} \sin {\gamma}_{ij}|$ on optic flow, which may be stated as the condition
\begin{equation}|q_{ij}|\leq \min\{|-\mu_{ij} n_{ij} \sin {\gamma}_{ij}|\}.
\label{e:qnoiseboundgenl}
\end{equation}

Assumption A1 provides $\min|\gamma_{ij}|\geq\Gamma,$ and A2 provides $\min|\mu_{ij}|\geq1/\rho,$ thus under A1 and A2, Eqn.~\eqref{e:qnoiseboundgenl} is bounded. Choose $n_{ij}=\bar{n}$, where $\bar{n}$ is the maximum noise magnitude on the speed measurement for which a CS dynamics convergence guarantee is available, then we can obtain $\bar{q}$, an upper bound of noise on the optic flow signal for which the YFM is convergent as  $\bar{q}=\frac{\bar{n}\sin{\Gamma}}{\rho}.$
\end{proof}

\paragraph{Dimensionality} The C-S model was originally developed describing 3-dimensional Euclidean space \citep{cucker2007TacEmergent}. This work applies the model to planar rigid body motion also having 3 configuration variables, and the no sideslip condition results in $\theta$ being fully specified by the velocity.

\paragraph{Required measurements} The primary sensory inputs in this method are optic flow $\dot{Q}$ and rotation rate $\dot{\theta}$. Both insects and robotic implementations have relatively mature pathways to recover these signals (compound eyes/cameras and gyros/halteres). While compass heading is used, it is applied in a quadrant level filter, and thus a relatively imprecise compass heading is sufficient, which is easily recoverable from even a rudimentary solar compass. When more precise measurements of $\theta_i$ are available, an orientation stiffness may be integrated into Eqn.~\eqref{e:thetaDot}, suggesting one reason insects in outdoor conditions may outperform swarm attempts in indoor conditions.

This modest sensor requirement is a significant advance from historic frameworks like C-S requiring the relative position and velocity of all other agents, which would require either an explicit communication network with sufficient bandwidth for the number of agents (a poorly scaled problem) or onboard sensing of those quantities, which is both prohibitive for small unmanned aerial systems and has no clear sensory path in insect biology.

\paragraph{Use of cotangent} $\cot\alpha$ is poorly behaved at zero and $\pi$. $\alpha=0$ does not occur in this problem by construction.  Its presence in the denominator of inputs $u_i^v$ and $u_i^\omega$ has the attractive feature that small values of $\alpha$ result in a vanishing contribution to the control, enforcing in a decaying response with distance. Along with the minimum resolvable angle $\munderbar{\alpha}$ and occlusal limit $\Gamma$, the $\cot^2\alpha$ scaling serves to mitigate the field of view, occlusion, and distance challenges of visual sensing described in previous work \citep{asadi2016limitedFovSwarm,soria2019LatVisionFlocking}.

\paragraph{$\alpha$ and noise}
In a biological example, $\alpha$ may be identified from $\dot{Q}$, e.g., it is sufficient for an agent to count the number of elementary motion detectors outputting a neighbor-induced optic flow \citep{haag2004reichardtDetectors}. The simulations directly applying noise on $\alpha$ do demonstrate sensitivity. However, noise on the primary measurement $\dot{Q}$ does not immediately imply noise on $\alpha$ as regions of large and small signal remain uniquely defined for $n_q/\dot{Q}$ noise-to-signal ratios below 1. 

\paragraph{Applicability} The applicability to $(v,r)$ flocking structures in previous literature provides a broad applicability to numerous models. For example, the use of YFM on the C-S dynamics can be seen as a systematic mapping updating both speed and heading. When the framework is applied to more restricted examples like Vicsek models incorporating heading only changes, it is comparable to feedback laws providing visually guided headings updates like \citet{moshtagh2007flockingTac}, whose constant (unit) speed and evolvable heading construction could be thought of as a visually-guided analog to a Vicsek model \citep{vicsek1995swarm}. The framework in this paper shows that if an idealized swarming (perfect information) feedback rule may be written in $(v,r)$ form, a visual navigation adaptation of those models exists, creating a foundation to systematically compare existing and future $(v,r)$ models \citep{reynolds1987flocksBoids,vicsek1995swarm,cucker2007TacEmergent}.

\section{Conclusion}\label{Conclusion}
This study developed a framework for mapping distributed feedback flocking update rules into visually guided swarming law that do not require a position reference or explicit communication network, and is suitable for distributed feedback on aerial systems. The approach relies on optic flow to ensure smooth flocking.  The primary proof relies on creating a correspondence of visual navigation to the accepted Cucker-Smale proof, in which instead of velocities and position vectors, we have used optic flow fields. The results are then further generalized to show that agent size need not be known, the algorithm provides some robustness to measurement noise, and has some differences in behavior to traditional Cucker-Smale. The impact of this framework is to provide concise, visually-guided analogues to a class of models able to be expressed in velocity and radius form, such as the well-known \citet{vicsek1995swarm} and \citet{reynolds1987flocksBoids}, improving the relevance of these models to robotic systems and to understanding the behavior of insect paths.

\nolinenumbers 
\bibliographystyle{IEEEtranN} 
\bibliography{00YMFshortPreprint.bbl}

\begin{thebibliography}{37}
\providecommand{\natexlab}[1]{#1}
\providecommand{\url}[1]{#1}
\csname url@samestyle\endcsname
\providecommand{\newblock}{\relax}
\providecommand{\bibinfo}[2]{#2}
\providecommand{\BIBentrySTDinterwordspacing}{\spaceskip=0pt\relax}
\providecommand{\BIBentryALTinterwordstretchfactor}{4}
\providecommand{\BIBentryALTinterwordspacing}{\spaceskip=\fontdimen2\font plus
\BIBentryALTinterwordstretchfactor\fontdimen3\font minus
  \fontdimen4\font\relax}
\providecommand{\BIBforeignlanguage}[2]{{%
\expandafter\ifx\csname l@#1\endcsname\relax
\typeout{** WARNING: IEEEtranN.bst: No hyphenation pattern has been}%
\typeout{** loaded for the language `#1'. Using the pattern for}%
\typeout{** the default language instead.}%
\else
\language=\csname l@#1\endcsname
\fi
#2}}
\providecommand{\BIBdecl}{\relax}
\BIBdecl

\bibitem[Parsons et~al.(2010)Parsons, Krapp, and
  Laughlin]{parsons2010sensorFusion}
M.~M. Parsons, H.~G. Krapp, and S.~B. Laughlin, ``Sensor fusion in identified
  visual interneurons,'' \emph{Current Biology}, vol.~20, no.~7, pp. 624--628,
  2010.

\bibitem[Taylor and Krapp(2007)]{taylor2007whatMeasureAndWhy}
G.~K. Taylor and H.~G. Krapp, ``Sensory systems and flight stability: what do
  insects measure and why?'' \emph{Advances in insect physiology}, vol.~34, pp.
  231--316, 2007.

\bibitem[Humbert and Hyslop(2009)]{humbert2009visuomotorConvergence}
J.~S. Humbert and A.~M. Hyslop, ``Bioinspired visuomotor convergence,''
  \emph{IEEE Transactions on Robotics}, vol.~26, no.~1, pp. 121--130, 2009.

\bibitem[Billah and Faruque(2020)]{billah2020multiAgentWfi}
M.~A. Billah and I.~A. Faruque, ``Bioinspired visuomotor feedback in a
  multiagent group/swarm context,'' \emph{IEEE Transactions on Robotics},
  vol.~37, no.~2, pp. 603--614, 2020.

\bibitem[Liu and Passino(2006)]{passino2006InformationFlow}
Y.~Liu and K.~M. Passino, ``Cohesive behaviors of multiagent systems with
  information flow constraints,'' \emph{IEEE Transactions on Automatic
  Control}, vol.~51, no.~11, pp. 1734--1748, 2006.

\bibitem[Yıldız and Özgüler(2015)]{bulent2015partialInfoNashEqb}
A.~Yıldız and A.~B. Özgüler, ``Partially informed agents can form a swarm
  in a nash equilibrium,'' \emph{IEEE Transactions on Automatic Control},
  vol.~60, no.~11, pp. 3089--3094, 2015.

\bibitem[Chung et~al.(2018)Chung, Paranjape, Dames, Shen, and
  Kumar]{chung2018SwarmSurveyTac}
S.-J. Chung, A.~A. Paranjape, P.~Dames, S.~Shen, and V.~Kumar, ``A survey on
  aerial swarm robotics,'' \emph{IEEE Transactions on Robotics}, vol.~34,
  no.~4, pp. 837--855, 2018.

\bibitem[Rossi et~al.(2018)Rossi, Bandyopadhyay, Wolf, and
  Pavone]{rossi2018review}
F.~Rossi, S.~Bandyopadhyay, M.~Wolf, and M.~Pavone, ``Review of multi-agent
  algorithms for collective behavior: a structural taxonomy,''
  \emph{IFAC-PapersOnLine}, vol.~51, no.~12, pp. 112--117, 2018.

\bibitem[Olfati-Saber(2006)]{olfatisaber2006Review}
R.~Olfati-Saber, ``Flocking for multi-agent dynamic systems: algorithms and
  theory,'' \emph{IEEE Transactions on Automatic Control}, vol.~51, no.~3, pp.
  401--420, 2006.

\bibitem[Reynolds(1987)]{reynolds1987flocksBoids}
C.~W. Reynolds, ``Flocks, herds and schools: A distributed behavioral model,''
  in \emph{Proceedings of the 14th annual conference on Computer graphics and
  interactive techniques}, 1987, pp. 25--34.

\bibitem[Vicsek et~al.(1995)Vicsek, Czir{\'o}k, Ben-Jacob, Cohen, and
  Shochet]{vicsek1995swarm}
T.~Vicsek, A.~Czir{\'o}k, E.~Ben-Jacob, I.~Cohen, and O.~Shochet, ``Novel type
  of phase transition in a system of self-driven particles,'' \emph{Physical
  review letters}, vol.~75, no.~6, p. 1226, 1995.

\bibitem[Cucker and Smale(2007)]{cucker2007TacEmergent}
F.~Cucker and S.~Smale, ``Emergent behavior in flocks,'' \emph{IEEE
  Transactions on automatic control}, vol.~52, no.~5, pp. 852--862, 2007.

\bibitem[Choi et~al.(2017)Choi, Ha, and Li]{Choi2017csModelAndVariants}
Y.~P. Choi, S.~Y. Ha, and Z.~Li, ``{Emergent dynamics of the cucker–Smale
  flocking model and its variants},'' in \emph{Active Particles, Volume 1 :
  Advances in Theory, Models, and Applications}, N.~Bellomo, P.~Degond, and
  E.~Tadmor, Eds.\hskip 1em plus 0.5em minus 0.4em\relax Cham: Springer
  International Publishing, 2017, vol.~1, no. 9783319499949, pp. 299--331.

\bibitem[Cucker and Mordecki(2008)]{Cucker2008mordeckiNoise}
F.~Cucker and E.~Mordecki, ``{Flocking in noisy environments},'' \emph{Journal
  des Mathematiques Pures et Appliquees}, vol.~89, no.~3, pp. 278--296, 2008.

\bibitem[Jabin and Wang(2017)]{Jabin2017meanFieldStochastic}
P.~E. Jabin and Z.~Wang, ``{Mean field limit for stochastic particle
  systems},'' \emph{Modeling and Simulation in Science, Engineering and
  Technology}, no. 9783319499949, pp. 379--402, 2017.

\bibitem[Park(1981)]{Park1981wienerProcess}
C.~Park, ``{Representations of gaussian processes by wiener processes},''
  \emph{Pacific Journal of Mathematics}, vol.~94, no.~2, pp. 407--415, 1981.

\bibitem[Tang et~al.(2019)Tang, Hu, Cui, Liao, Lao, Lin, and
  Teo]{2019yazheVisionAidedFlocking}
Y.~Tang, Y.~Hu, J.~Cui, F.~Liao, M.~Lao, F.~Lin, and R.~S.~H. Teo,
  ``Vision-aided multi-uav autonomous flocking in gps-denied environment,''
  \emph{IEEE Transactions on Industrial Electronics}, vol.~66, no.~1, pp.
  616--626, 2019.

\bibitem[Moshtagh and Jadbabaie(2007)]{moshtagh2007flockingTac}
N.~Moshtagh and A.~Jadbabaie, ``Distributed geodesic control laws for flocking
  of nonholonomic agents,'' \emph{IEEE Transactions on Automatic Control},
  vol.~52, no.~4, pp. 681--686, 2007.

\bibitem[Dachner et~al.(2022)Dachner, Wirth, Richmond, and
  Warren]{dachner2018humanVisualCoupling}
\BIBentryALTinterwordspacing
G.~C. Dachner, T.~D. Wirth, E.~Richmond, and W.~H. Warren, ``The visual
  coupling between neighbours explains local interactions underlying human
  \&\#x2018;flocking','' \emph{Proceedings of the Royal Society B: Biological
  Sciences}, vol. 289, no. 1970, p. 20212089, 2022. [Online]. Available:
  \url{https://royalsocietypublishing.org/doi/abs/10.1098/rspb.2021.2089}
\BIBentrySTDinterwordspacing

\bibitem[Asadi et~al.(2016)Asadi, Ajorlou, and
  Aghdam]{asadi2016limitedFovSwarm}
\BIBentryALTinterwordspacing
M.~M. Asadi, A.~Ajorlou, and A.~G. Aghdam, ``Distributed control of a network
  of single integrators with limited angular fields of view,''
  \emph{Automatica}, vol.~63, pp. 187--197, 2016. [Online]. Available:
  \url{https://www.sciencedirect.com/science/article/pii/S0005109815003982}
\BIBentrySTDinterwordspacing

\bibitem[Soria et~al.(2019)Soria, Schiano, and
  Floreano]{soria2019LatVisionFlocking}
E.~Soria, F.~Schiano, and D.~Floreano, ``The influence of limited visual
  sensing on the reynolds flocking algorithm,'' in \emph{2019 Third IEEE
  International Conference on Robotic Computing (IRC)}, 2019, pp. 138--145.

\bibitem[Serres and Ruffier(2017)]{serres2017opticFlowCollision}
J.~R. Serres and F.~Ruffier, ``Optic flow-based collision-free strategies: From
  insects to robots,'' \emph{Arthropod structure \& development}, vol.~46,
  no.~5, pp. 703--717, 2017.

\bibitem[Krapp and Hengstenberg(1996)]{krapp1996ofMeasurement}
H.~G. Krapp and R.~Hengstenberg, ``Estimation of self-motion by optic flow
  processing in single visual interneurons,'' \emph{Nature}, vol. 384, no.
  6608, pp. 463--466, 1996.

\bibitem[Nordstr{\"o}m and O’Carroll(2009)]{nordstrom2009stmdHypercomplex}
K.~Nordstr{\"o}m and D.~C. O’Carroll, ``Feature detection and the
  hypercomplex property in insects,'' \emph{Trends in neurosciences}, vol.~32,
  no.~7, pp. 383--391, 2009.

\bibitem[Nordstr{\"o}m(2012)]{nordstrom2012stmdsNeuralSpecializations}
K.~Nordstr{\"o}m, ``Neural specializations for small target detection in
  insects,'' \emph{Current opinion in neurobiology}, vol.~22, no.~2, pp.
  272--278, 2012.

\bibitem[Wiederman et~al.(2017)Wiederman, Fabian, Dunbier, and
  O’Carroll]{wiederman2017predictiveFocusGainModulation}
S.~D. Wiederman, J.~M. Fabian, J.~R. Dunbier, and D.~C. O’Carroll, ``A
  predictive focus of gain modulation encodes target trajectories in insect
  vision,'' \emph{Elife}, vol.~6, p. e26478, 2017.

\bibitem[Seelig and Jayaraman(2015)]{seelig2015neuralStructureCompass}
J.~D. Seelig and V.~Jayaraman, ``Neural dynamics for landmark orientation and
  angular path integration,'' \emph{Nature}, vol. 521, no. 7551, pp. 186--191,
  2015.

\bibitem[Zittrell et~al.(2020)Zittrell, Pfeiffer, and
  Homberg]{zittrell2020PolarizationSunCompassLocust}
F.~Zittrell, K.~Pfeiffer, and U.~Homberg, ``Matched-filter coding of sky
  polarization results in an internal sun compass in the brain of the desert
  locust,'' \emph{Proceedings of the National Academy of Sciences}, vol. 117,
  no.~41, pp. 25\,810--25\,817, 2020.

\bibitem[El~Jundi and Dacke(2021)]{dacke2021insectDrosophilaCompass}
B.~El~Jundi and M.~Dacke, ``Insect orientation: The drosophila wind compass
  pathway,'' \emph{Current Biology}, vol.~31, no.~2, pp. R83--R85, 2021.

\bibitem[Srinivasan(1992)]{srinivasan1992beesOpticFlow}
M.~V. Srinivasan, ``How bees exploit optic flow: behavioural experiments and
  neural models,'' \emph{Philosophical Transactions of the Royal Society of
  London. Series B: Biological Sciences}, vol. 337, no. 1281, pp. 253--259,
  1992.

\bibitem[Escobar-Alvarez et~al.(2019)Escobar-Alvarez, Ohradzansky, Keshavan,
  Ranganathan, and Humbert]{escobar2019smallObjectAvoidance}
H.~D. Escobar-Alvarez, M.~Ohradzansky, J.~Keshavan, B.~N. Ranganathan, and
  J.~S. Humbert, ``Bioinspired approaches for autonomous small-object detection
  and avoidance,'' \emph{IEEE Transactions on Robotics}, vol.~35, no.~5, pp.
  1220--1232, 2019.

\bibitem[Moshtagh et~al.()Moshtagh, Jadbabaie, and Daniilidis]{moshtaghcross}
N.~Moshtagh, A.~Jadbabaie, and K.~Daniilidis, ``Cross-product steering for
  distributed velocity alignment in multi-agent systems.''

\bibitem[Billah and Faruque(2022)]{billah2022stmdMultiAgent}
M.~A. Billah and I.~Faruque, ``The multi-agent group motions generated by
  models of insect small target detector neurons and feedback,'' in \emph{AIAA
  SCITECH 2022 Forum}, 2022, p. 0962.

\bibitem[Koenderink and van Doorn(1987)]{koenderink1987opticFlowFacts}
J.~J. Koenderink and A.~J. van Doorn, ``Facts on optic flow,'' \emph{Biological
  cybernetics}, vol.~56, no.~4, pp. 247--254, 1987.

\bibitem[Palm(2010)]{palm2010systemDynamics}
W.~J. Palm, \emph{System dynamics}.\hskip 1em plus 0.5em minus 0.4em\relax
  McGraw-Hill New York, 2010, vol.~2.

\bibitem[Sun and Lin(2015)]{Sun2015csNoiseMultPositive}
Y.~Sun and W.~Lin, ``A positive role of multiplicative noise on the emergence
  of flocking in a stochastic cucker-smale system,'' \emph{Chaos: An
  Interdisciplinary Journal of Nonlinear Science}, vol.~25, no.~8, p. 083118,
  2015.

\bibitem[Haag et~al.(2004)Haag, Denk, and Borst]{haag2004reichardtDetectors}
J.~Haag, W.~Denk, and A.~Borst, ``Fly motion vision is based on reichardt
  detectors regardless of the signal-to-noise ratio,'' \emph{Proceedings of the
  National Academy of Sciences}, vol. 101, no.~46, pp. 16\,333--16\,338, 2004.

\end{thebibliography}

\end{document}